\documentclass[11pt,a4paper,oneside,abstract]{scrartcl}

\pdfoutput=1
\usepackage[margin=1in,a4paper]{geometry}
\newif\iflong
\newif\ifshort

\longtrue

\iflong
\else
\shorttrue
\fi

\usepackage[utf8]{inputenc}
\usepackage[T1]{fontenc}
\usepackage[numbers,sort&compress]{natbib} %
\usepackage[breaklinks,pagebackref]{hyperref}%
\usepackage{mathtools}
\usepackage{amsmath,amssymb,amsthm}
\usepackage{tikz}
\usetikzlibrary{patterns}
\pgfdeclarelayer{background}
\pgfsetlayers{background,main}

\usepackage{microtype}
\usepackage[capitalize]{cleveref}
\usepackage{graphicx}
\usepackage{booktabs}
\usepackage{xcolor}
\usepackage[textsize=scriptsize,disable]{todonotes}
\usepackage{caption}
\usepackage{subcaption}
\usepackage{paralist}
\usepackage{verbatim}
\usepackage{xspace}
\iflong\else\usepackage[compact]{titlesec}\fi
\usepackage{authblk}

\addtokomafont{title}{\normalfont\rmfamily}
\addtokomafont{disposition}{\rmfamily}
\addtokomafont{descriptionlabel}{\rmfamily}

\hypersetup{allbordercolors=[HTML]{0404ee},pdfborderstyle={/S/U/W 1}}
\setdefaultitem{\textendash}{\textendash}{\textendash}{\textendash}
\setdefaultleftmargin{1em}{}{}{}{}{}
\setcapindent{0em}

\newcommand{\cp}{cy\-cle-pair}
\newcommand{\cv}{cy\-cle-ver\-tex}
\newcommand{\A}{A}
\newcommand{\B}{B}
\newcommand{\chld}{base}
\newcommand{\chlds}{bases}

\newcommand{\layers}{\ensuremath{r}}
\newcommand{\sepsize}{\ensuremath{p}}
\newcommand{\numseps}{\ensuremath{t}}
\newcommand{\Sep}{\ensuremath{S}}
\newcommand{\Sepseq}{\ensuremath{\mathcal S}}

\newcommand{\wfsf}{well-formed separator sequence}
\newcommand{\Wfsf}{Well-formed separator sequence}

\newcommand{\vcyca}{\ensuremath{v^*}}
\newcommand{\vcycb}{\ensuremath{v^\dag}}

\newtheorem{theorem}{Theorem}[section]
\theoremstyle{definition}
\newtheorem{observation}[theorem]{Observation}
\crefname{observation}{Observation}{Observations}
\newtheorem{construction}[theorem]{Construction}
\crefname{construction}{Construction}{Constructions}
\crefname{figure}{Figure}{Figures}
\crefname{case}{Case}{Cases}
\crefname{subcase}{Case}{Cases}

\newtheorem{definition}[theorem]{Definition}
\newtheorem{lemma}[theorem]{Lemma}
\newtheorem{corollary}[theorem]{Corollary}

\crefname{ssprop}{Condition}{Conditions}
\crefname{property}{Property}{Properties}
\creflabelformat{property}{(#2#1#3)}
\crefname{enumi}{Condition}{Conditions}
\creflabelformat{enumi}{(#2#1#3)}

\newtheoremstyle{case}%
{5pt plus 1pt minus 1pt}{5pt plus 1pt minus 1pt}%
{}{1em}%
{\itshape}{.}%
{5pt}%
{\thmname{#1}\thmnumber{ #2}: \thmnote{#3}}

\theoremstyle{case}
\newtheorem{case}{Case}
\newtheorem{subcase}{Case}
\newtheorem{claim}{Claim}
\numberwithin{subcase}{case}

\makeatletter
\@addtoreset{case}{theorem}
\@addtoreset{claim}{theorem}
\makeatother

\newcommand{\NP}{\text{\normalfont NP}}
\newcommand{\FPT}{\text{\normalfont FPT}}

\newcommand{\ie}{\emph{i.e.}}
\newcommand{\eg}{\emph{e.g.}}

\newcommand{\Hyp}{{\cal H}}
\newcommand{\Ver}{{\cal V}}
\newcommand{\Ed}{{\cal E}}
\newcommand{\nEd}{\ensuremath{m}}
\newcommand{\nVer}{\ensuremath{n}}
\newcommand{\R}{{\cal R}}
\newcommand{\T}{\ensuremath{\mathcal{T}}\xspace}
\renewcommand{\S}{\ensuremath{\mathcal{S}}\xspace}
\newcommand{\Ssubst}{\ensuremath{\mathcal{P}}\xspace}
\newcommand{\Tsubst}{\ensuremath{\mathcal{R}}\xspace}
\newcommand{\Srev}{\ensuremath{\mathcal{S}_{\text{rev}}}}

\newcommand{\G}{\ensuremath{\mathcal{G}}\xspace}
\newcommand{\poly}{\ensuremath{\operatorname{poly}}}

\newcommand{\sssl}{\ensuremath{\sqrt{(\log_\ell k + 1)/2} - 1}\xspace}
\newcommand{\sepseqlength}{\ensuremath{\lfloor\sqrt[2r]{\log n}/6^r\rfloor}\xspace}

\newcommand{\prob}[5]{%
  \begingroup
  \par\medskip
  \noindent \textsc{#1}\nopagebreak[4]
  \par\noindent\hangindent=\parindent\textit{#2}  #3
  \par\noindent\hangindent=\parindent\textit{#4}  #5
  \par  \medskip
  \endgroup
}

\newcommand{\parprob}[4]{%
  \begingroup
  \par\medskip
  \noindent \textsc{#1}\nopagebreak[4]
  \par\noindent\hangindent=\parindent\textit{Input:}  #2
  \par\noindent\hangindent=\parindent\textit{Question:}  #3
  \par\noindent\hangindent=\parindent\textit{Parameter:}  #4
  \par
  \endgroup
}

\DeclareMathOperator{\glue}{\circ}

\hypersetup{
  pdfauthor={René van Bevern, Iyad Kanj, Christian Komusiewicz, Rolf Niedermeier, Manuel Sorge},
  pdftitle={Well-Formed Separator Sequences, with an Application to Hypergraph Drawing},
  pdfkeywords={Kernelization, Planar Graphs, Hypergraph Supports, Layered Separation},
  pdflang={en}}

 \newcommand{\PS}{\textsc{Planar Support}}%
\newcommand{\HE}{\ensuremath{\mathcal{F}}}  \newcommand{\HV}{\ensuremath{V}}

\title{Well-Formed Separator Sequences, with an Application to Hypergraph Drawing\iflong{}\else{}\footnote{Due to space constraints, several details are deferred to a full version of this article (in the appendix). 
 René van Bevern, Iyad Kanj, and Manuel Sorge acknowledge support by the DFG, project DAPA (NI~369/12).  Parts of this work were done while René van Bevern was employed at TU Berlin and during a six month stay of Iyad Kanj at TU Berlin.}\fi}

\author[1]{René van Bevern}
\author[2]{Iyad Kanj}
\author[3]{Christian Komusiewicz}
\author[3]{Rolf Niedermeier}
\author[3]{Manuel Sorge}

\affil[1]{Novosibirsk State University, Novosibirsk, Russian Federation, \texttt{rvb@nsu.ru}}

\affil[2]{DePaul University, Chicago, USA, \texttt{ikanj@cs.depaul.edu}}

\affil[3]{Institut für Softwaretechnik und Theoretische Informatik, TU Berlin, Germany, \texttt{\{christian.komusiewicz, rolf.niedermeier, manuel.sorge\}@tu-berlin.de}}
\date{}

\begin{document}

\maketitle
\thispagestyle{empty}
\begin{abstract}\looseness=-1
\noindent Given a hypergraph~$\Hyp$, the \PS{} problem asks whether there is a planar graph~$G$ on the same vertex set as~$\Hyp$ such that each hyperedge
induces a connected subgraph of~$G$. \PS{} is motivated by
applications in graph drawing and data visualization.
We show that \PS{} is fixed-parameter tractable  when parameterized
by the number of hyperedges in the input hypergraph and the outerplanarity number of the sought planar graph.  To this end, we develop novel structural results for $r$-outerplanar triangulated disks,
showing that they admit sequences of separators with structural properties enabling data reduction. This allows us to obtain a problem kernel for \PS{}, thus showing its fixed-parameter tractability.
\end{abstract}

\iflong{}
\fi{}
\section{Introduction}
\label{sec:intro}
\looseness=-1 A \emph{support} for a hypergraph~$\Hyp=(\Ver,\Ed)$ is a graph $G$ on the same vertex set~$\Ver$ such that, for each hyperedge~$e\in \Ed$, the subgraph of $G$ induced by the vertices in~$e$ is connected.  If there is no restriction on the support, then any given hypergraph $\Hyp=(\Ver,\Ed)$ has a support, namely the clique on~$\Ver$.  For a graph property $\Pi$, the problem of deciding whether a given hypergraph~$\Hyp$ has a support that satisfies $\Pi$---shortly, a $\Pi$-support---has been studied by various research communities for numerous properties~$\Pi$.
This problem has, among others, applications in graph drawing, databases, and social and overlay networks~\cite{aar10,CMTV07,CKNSSW14, dum88, KKS08,OR11,bkmsv,BFMY83,JP87,TY84,BCPS11}. The studied graph properties include: having minimum number of edges, being a path, a cycle, a tree, having bounded treewidth, being planar, and being $r$-outerplanar. For some of these properties, the problem is known to be solvable in polynomial time (\eg, path~\cite{ks2003,bkmsv}, cycle~\cite{bkmsv}, tree~\cite{BFMY83,JP87,TY84}), for some it is known to be \NP-hard (\eg, minimum number of edges~\cite{du86}, planar~\cite{JP87}, 2-outerplanar~\cite{bkmsv}), and for some its complexity remains unresolved (\eg, outerplanar~\cite{bkmsv}).

\paragraph{Planar supports.} \looseness=-1 Perhaps the majority of the work on hypergraph support problems is related to hypergraph drawing or representation.
Here, one seeks a plane drawing of the hypergraph that captures the relations among its vertices---stipulated by its hyperedges, while revealing these relations elegantly via the drawing of the hypergraph. One method for drawing hypergraphs is to draw them as \emph{vertex-based Venn diagrams}~\cite{JP87}, %
also referred to as \emph{subdivision drawings}~\cite{KKS08}. A subdivision drawing of a hypergraph is a plane subdivision such that each vertex of the hypergraph corresponds uniquely to a face of the subdivision, and for each hyperedge the union of all the faces corresponding to the vertices in the hyperedge forms a connected region. A~hypergraph has a subdivision drawing if and only if it has a planar support~\cite{KKS08}. Deciding whether a hypergraph has a planar support is \NP-complete~\cite{JP87}.
We study the parameterized complexity of the problem %
parameterized by the combination of the number of hyperedges in the hypergraph and the outerplanarity~$r$ of the sought support:%
\looseness=-1
\parprob{\PS}
{A hypergraph~$\Hyp$ with \nVer{}~vertices and $\nEd{}$~hyperedges, and an $r \in \mathbb{N}$.}
{Does $\Hyp$ have a planar support of outerplanarity at most~$r$?}
{The number~$\nEd{}$ of hyperedges in~$\Hyp$ and~$r$ combined.}
\paragraph{Known results.}\looseness=-1 \PS{} is \NP-hard for~$r=2$~\cite{bkmsv}. \citet{bkmsv} give an \NP-hardness reduction for~$r=3$\iflong{} that transforms a 3-SAT instance into a hypergraph~$\Hyp$.  By inspecting the construction of $\Hyp$, it can be easily verified that $\Hyp$ either has a 3-outerplanar support or no planar support at all. Thus, the reduction of \citet{bkmsv} implies that, for every fixed~$r>3$, \PS~is \NP-hard as well: in the reduction, simply add a fixed~$r$-outerplanar graph to~$\Hyp$; the resulting hypergraph has an~$r$-outerplanar support if and only if~$\Hyp$ has a 3-outerplanar support.\fi{} The (classical) complexity of \PS~for $r=1$ remains open~\cite{bkmsv}.

An underlying assumption for several results in the literature pertaining to \PS~(\eg, \citet[p. 179]{makinen}, \citet[p. 346]{bkmsv}, \citet[p. 399]{KKS08}) has been that the hypergraph is \emph{twinless}, that is, does not contain two vertices (\emph{twins}) such that the set of hyperedges containing the first is the same as that containing the second. (Twins were referred to as ``equivalent'' vertices by \citet{bkmsv}.)
The intuition behind this assumption is that a twin does not affect the instance because whatever can be ``achieved'' by a vertex can be achieved by its twin.\iflong{} In \cref{sec:supp}, we\fi{} demonstrate that this assumption changes the landscape of \PS~completely: \iflong{}\fi{}we exhibit hypergraphs with twins that admit ($r$-outer)planar supports but depriving them of their twins results in hypergraphs with no ($r$-outer)planar supports.
This illustrates the important role that twins play in realizing ($r$-outer)planar supports for hypergraphs. Indeed, the presence of twins makes \PS~much more challenging: one can easily show \PS~for twinless hypergraphs to be fixed-parameter tractable (\FPT), whereas showing \FPT{} in general hypergraphs is much more demanding.

\looseness=-1 It can be shown that, for each value of the parameter~$\nEd{}$, there is an (unknown and distinct) algorithm solving \PS{} in $f(\nEd{})\cdot\poly(n)$~time for some function~$f$. %
In other words, \PS{} is \emph{non-uniformly} \FPT{} when parameterized by~$\nEd{}$ and hence, also when parameterized by~$\nEd{}$ and~$r$ combined.  To obtain this result, one can use the known machinery of \emph{well-quasi orderings}~\cite{DF13} to prove that each yes-instance of \PS{} contains some minimal yes-instance, and that the number and size of minimal yes-instances depends only on~$\nEd{}$\iflong{}. We outline the proof in \cref{sec:nonuni-fpt}.\fi{} Note, however, that even undecidable problems can be non-uniformly \FPT{}.  Therefore, non-uniform \FPT{} results in general are unimplementable and extensive research has focused on making non-uniform \FPT{} results uniform (\eg, numerous \FPT{} results that can be obtained using graph minor theory).

\paragraph{Our contributions.} We present an algorithm that decides a given instance of \PS{} in $f(\nEd{},r)\cdot\poly(n)$~time, where $f$~is \emph{explicitly given}. This implies that \PS{} is strongly uniformly \FPT{}, which is a strong improvement over the above-mentioned non-uniform \FPT{} result, and a necessary step in order to get applicable algorithms.

\looseness=-1 We prove that \PS{} is \FPT{} by providing a problem kernel. Notably, the number of hyperedges is perhaps the most natural parameter to study, as already observed in previous work~\cite{CKNSSW14,HHIOSW12}. Also note that a problem kernel with respect to the studied parameter combination is a stronger result than having  for each fixed~$r$ a problem kernel with respect to the parameter ``number of hyperedges''.  The main ingredient of the problem kernel is the non-trivial observation that, indeed, removing one of \emph{sufficiently many} twins does not affect the instance.
To obtain the problem kernel, based on the crucial observation that, without loss of generality, we can focus on triangulated disks, we prove a general structural result about separators of $r$-outerplanar triangulated disks. This is of independent interest: Given an embedding of an $r$-outerplanar triangulated disk $G$ ($r\geq 1$) on \nVer{}~vertices, we show that one can construct in polynomial time a sequence of separators for~$G$, which we refer to as a \emph{\wfsf{}} and whose length is some increasing, unbounded function in~$r$ and~$n$. %

\looseness=-1 We formally introduce \wfsf{}s in \cref{sec:wfsf} and compare them to other separator families found in the literature. Their structural properties make them amenable to a \emph{gluing} operation, which is also introduced in \cref{sec:wfsf}.  Gluing removes the subgraph of~$G$ ``between'' any two separators in the embedding, and identifies the separators.  We show that gluing any two separators in a \wfsf{} preserves the $r$-outerplanarity of~$G$. To apply this toolkit to \PS, we show that if the number of vertices in the hypergraph $\Hyp$ is ``large'' with respect to the parameter, then there are two separators in the planar support (if one exists) such that the subgraph between the two separators is ``redundant''\iflong{} (\ie, does not have any effect on the connectivity of the hyperedges)\fi, a property that we capture using the notion of \emph{separator signatures} (\cref{sec:application}). The above allows us to conclude that if an $r$-outerplanar support for $\Hyp$ exists, then a support whose size is upper-bounded by a function of the parameter must exist as well. This gives a problem kernel and, as a consequence, an \FPT~algorithm for \PS.  \cref{sec:algo} provides the technical construction of well-formed separator sequences.

\section{Preliminaries}
\label{sec:prelim}
We use standard terminology from graph theory~\citep{west} and parameterized complexity~\citep{DF13,FG06,rolfbook}.
\iflong{}%
\paragraph{Graphs.}
\fi{} Unless stated otherwise, all graphs are without parallel edges or loops. A~\emph{cut-vertex} (resp.\ \emph{cut-edge}) in a connected graph~$G$ is a vertex~$v$ (resp. an edge~$e$) such that $G-v$ (resp.\ $G-e$) is disconnected.  A~connected graph~$G$ is \emph{biconnected} if no vertex in~$G$ is a cut-vertex.  The \emph{blocks} of a graph~$G$ are its maximal biconnected subgraphs, its cut-edges, and its isolated vertices. %
\iflong{}
\paragraph{\boldmath $r$-Outerplanar disks.}
\fi{}
A \emph{plane graph}~$G = (V, E)$ is a planar graph given with a fixed embedding in the plane. The \emph{layer decomposition} of~$G$ with respect to the embedding is a partition of~$V$ into layers~$L_1\uplus\dots\uplus L_r$ is defined inductively as follows. Layer~$L_1$ is the set of vertices that lie on the outer face of~$G$, and layer~$L_i$ is the set of vertices that lie on the outer face of~$G - \bigcup_{j=1}^{i-1} L_j$ for $ 1 < i \leq r$. The graph $G$ is called \emph{$r$-out\-er\-pla\-nar} if it has an embedding with a layer decomposition consisting of at most $r$~layers. If $r=1$, then $G$~is simply said to be \emph{outerplanar}.
A plane graph~$G$ is said to be \emph{triangulated} if each face of~$G$, including the outer face, is a triangle, and $G$
is said to be a \emph{triangulated disk} if its outer face is a simple cycle (not necessarily a triangle), and all its inner faces are triangles~\cite{Biedl15}.  It is easy to see that the vertices on the outer face of a biconnected $r$-outerplanar graph form a simple cycle.
In most sections of this paper, we will be working with a fixed $r$-outerplanar triangulated disk~$G$, that is, we implicitly fix an embedding of~$G$. When the context is clear, we will often abuse the notation and use
$L_1$ to refer to the simple cycle that delimits the outer face of~$G$. It is known that any vertex~$v$ in layer~$L_i$, $i > 1$, of an $r$-outerplanar triangulated disk~$G$ has a neighbor in layer $L_{i-1}$~\cite{Biedl15}.

\iflong{}
\paragraph{Hypergraphs.}
\fi{} A \emph{hypergraph} $\Hyp = (\Ver,\Ed)$ consists of a vertex set $\Ver = V({\cal H})$ and an edge set~$\Ed = E({\cal H})$ such that $e \subseteq \Ver$ for every $e \in \Ed$.  Throughout this work, we denote $\nVer{}:=|\Ver|$ and $\nEd{}:=|\Ed|$.  The \emph{size} of a hyperedge is the number of vertices in it.  Unless stated otherwise, we assume that hypergraphs do not contain hyperedges of size at most 1 or multiple copies of the same hyperedge. %
For a vertex $v \in \Hyp$, we denote $\Ed(v):=\{e \in {\cal H} \mid v \in e\}$.  A~vertex~$v$ \emph{covers} a vertex~$u$ if $\Ed(u) \subseteq \Ed(v)$. Two vertices $u, v \in \Ver$ are \emph{twins} if $\Ed(v)=\Ed(u)$. Clearly, the relation~$\R$ on~$\Ver$ defined by $\forall u, v \in \Ver, u \R v\iff\Ed(u)=\Ed(v)$ is an equivalence relation.
We write $[u]_{\R}$ to denote the \emph{twin class} of a vertex~$u \in \Ver$ under the above relation~$\R$.  \emph{Removing a vertex set}~$S$ from a hypergraph~$\Hyp = (\Ver,\Ed)$ results in the hypergraph~$\Hyp-S:=(\Ver\setminus S, \Ed')$ where~$\Ed'$ is obtained from $\{e\setminus S\mid e\in \Ed\}$ by removing the empty set and singleton sets. We use~$\Hyp[S]:=\Hyp-(\Ver\setminus S)$ and $\Hyp-v:=\Hyp-\{v\}$.

\iflong{}
 \paragraph{Parameterized complexity.}
 A \emph{parameterized problem} is a set of instances of the form $({\cal I}, k)$, where
 ${\cal I} \in \Sigma^*$ for a finite alphabet~$\Sigma$, and $k \in \mathbb{N}$ is the \emph{parameter}.
 A parameterized problem $Q$ is \emph{fixed-parameter tractable}, shortly \FPT, if there exists an algorithm that on input
 $({\cal I}, k)$ decides if $({\cal I}, k)$ is a yes-instance of~$Q$ in $f(k)|{\cal I}|^{O(1)}$ time,
 where $f$ is a computable function independent of $|{\cal I}|$.
 A parameterized problem $Q$ is \emph{kernelizable}
 if there exists a polynomial-time self-reduction that maps an instance $({\cal I},k)$ of
 $Q$ to another instance $({\cal I}',k')$ of $Q$ such that: (1) $|{\cal I}'| \leq \lambda(k)$ for
 some computable function $\lambda$, (2) $k' \leq \lambda(k)$, and (3) $({\cal I},k)$ is a yes-instance
 of $Q$ if and only if $({\cal I}',k')$ is a yes-instance of $Q$. The instance
 $({\cal I}',k')$ is called the \emph{problem kernel} of $({\cal I}, k)$.
 It is well known that a parameterized problem is fixed-parameter tractable if and only if
 the problem is kernelizable. %
\fi{}

\section{\Wfsf{}s}\label{sec:wfsf}
\tikzstyle {sepa} = [line width=14pt, color=red, opacity=0.5, line cap=round, line join=round]
\tikzstyle {sepb} = [line width=11pt, color=green,opacity=0.5, line cap=round, line join=round]
\tikzstyle {sepc} = [line width=8pt, color=blue, opacity=0.5, line cap=round, line join=round]
\tikzstyle {vert} = [circle, thick, fill=white, draw, inner sep=0pt, minimum size=2mm]
\tikzstyle {layer} = [dashed]

\begin{figure}[t]
   \begin{tikzpicture}[y=0.35cm,x=0.75cm,baseline=(current bounding box.center)]
     \node[vert] (l3a) at (0,1) {};
     \node[vert] (l3b) at (2,1) {};
     \node[vert] (l3c) at (2,0) {};
     \node[vert] (l3d) at (1,0) {};
     \node[vert] (l3e) at (0,0) {};

     \node[vert] (l2a) at (0,3) {};
     \node[vert] (l2b) at (2,3) {};
     \node[vert] (l2c) at (2,-2) {};
     \node[vert] (l2d) at (1,-2) {};
     \node[vert] (l2e) at (0,-2) {};

     \node[vert] (l1a) at (0,5) {};
     \node[vert] (l1a') at (1,5) {};
     \node[vert] (l1b) at (2,5) {};
     \node (l1c) at (2,-4) {};
     \node[vert] (l1d) at (1,-4) {};
     \node[vert] (l1e) at (0,-4) {};
     \begin{pgfonlayer}{background}
       \draw[layer] (l3a)--(l3b) to[bend left=90] (l3c)--(l3d)--(l3e) to[bend left=90] (l3a);

       \draw[layer] (l2a)--(l2b) to[bend left=90] (l2c)--(l2d)--(l2e) to[bend left=90] (l2a);

       \draw[layer] (l1a)--(l1b) to[bend left=90] (l1c)--(l1d)--(l1e) to[bend left=90] (l1a);

       \draw[sepa] (l1a.center) -- (l2a.center) -- (l3a.center) -- (l3e.center) -- (l2e.center) -- (l1e.center); \draw[sepb] (l1a'.center)--(l2a.center)--(l3a.center)--(l3d.center)--(l2d.center)--(l1d.center); \draw[sepc] (l1b.center)--(l2b.center)--(l3b.center)--(l3c.center)--(l2c.center)--(l1d.center);

       \draw (l1a.center) -- (l2a.center) -- (l3a.center) -- (l3e.center) -- (l2e.center) -- (l1e.center); \draw (l1a'.center)--(l2a.center)--(l3a.center)--(l3d.center)--(l2d.center)--(l1d.center); \draw (l1b.center)--(l2b.center)--(l3b.center)--(l3c.center)--(l2c.center)--(l1d.center);
     \end{pgfonlayer}
   \end{tikzpicture}\hfill
   \begin{tikzpicture}[y=0.35cm,x=0.75cm,baseline=(current bounding box.center)]
     \node[vert] (l3a) at (0,1) {};
     \node[vert] (l3b) at (2,1) {};
     \node[vert] (l3c) at (2,0) {};
     \node[vert] (l3d) at (1,0) {};
     \node[vert] (l3e) at (0,0) {};

     \node[vert] (l2a) at (0,3) {};
     \node[vert] (l2b) at (2,2) {};
     \node[vert] (l2b') at (1,2.5) {};
     \node (l2c) at (2,-1) {};
     \node[vert] (l2d) at (1,-1.5) {};
     \node[vert] (l2e) at (0,-2) {};

     \node[vert] (l1') at (3.5,0.5) [label=right:$v^*$] {};
     \node (l1a) at (0,4) {};
     \node (l1a') at (1,4) {};
     \node (l1b) at (2,4) {};
     \node (l1c) at (2,-3) {};
     \node (l1d) at (1,-3) {};
     \node (l1e) at (0,-3) {};

     \node (l0a) at (0,5) {};
     \node (l0b) at (3.25,5) {};
     \node (l0c) at (3.25,-4) {};
     \node (l0d) at (0,-4) {};

     \draw[layer] (l0a.center)--(l0b.center) to[bend left=90] (l0c.center)--(l0d.center) to[bend left=90] (l0a.center);

     \begin{pgfonlayer}{background}
       \draw[layer] (l3a)--(l3b) to[bend left=90] (l3c)--(l3d)--(l3e) to[bend left=90] (l3a);

       \draw[layer] (l2a.center)--(l2b.center) to[bend left=90] (l2c.center)--(l2d.center)--(l2e.center) to[bend left=70] (l2a.center);

       \draw[layer] (l1a.center)--(l1b.center) to[out=0,in=90] (l1'.center) to[out=-90,in=0] (l1c.center) -- (l1d.center) -- (l1e.center) to[bend left=70] (l1a.center);

       \draw[sepa] (l1'.center) to[out=90,in=30] (l2a.center)--(l3a.center)--(l3e.center)--(l2e.center) to[out=-30,in=-90](l1'.center);

       \draw[sepb] (l1'.center) to[out=100,in=30] (l2b'.center) -- (l3a.center) -- (l3d.center) -- (l2d.center) to[out=-30,in=-100] (l1'.center);

       \draw[sepc] (l1'.center) to[out=120, in=0] (l2b.center) -- (l3b.center) -- (l3c.center) --(l2d.center) to[out=-30,in=-100](l1'.center);
       \draw (l1'.center) to[out=90,in=35] (l2a.center)--(l3a.center)--(l3e.center)--(l2e.center) to[out=-35,in=-90](l1'.center);

       \draw (l1'.center) to[out=100,in=30] (l2b'.center) -- (l3a.center) -- (l3d.center) -- (l2d.center) to[out=-30,in=-100] (l1'.center);

       \draw (l1'.center) to[out=120, in=0] (l2b.center) -- (l3b.center) -- (l3c.center) --(l2d.center);
     \end{pgfonlayer}
   \end{tikzpicture}\hfill
   \begin{tikzpicture}[x=0.75cm,y=0.35cm,baseline=(current bounding box.center)]
     \node[vert] (l3a) at (0,1) {};
     \node[vert] (l3b) at (2,1) {};
     \node[vert] (l3c) at (2,0) {};
     \node[vert] (l3d) at (1,0) {};
     \node[vert] (l3e) at (0,0) {};

     \node[vert] (l2a) at (0,3) {};
     \node[vert] (l2b) at (2,2) {};
     \node[vert] (l2b') at (1,2.5) {};
     \node (l2c) at (2,-1) {};
     \node[vert] (l2d) at (1,-1.5) {};
     \node[vert] (l2e) at (0,-2) {};

     \node[vert] (l1') at (3,2) [label=right:$v^*$] {};
     \node[vert] (l1'') at (3,-1) [label=right:$v^\dag$] {};
     \node (l1a) at (0,4) {};
     \node (l1a') at (1,4) {};
     \node (l1b) at (2,4) {};
     \node (l1c) at (2,-3) {};
     \node (l1d) at (1,-3) {};
     \node (l1e) at (0,-3) {};

     \node (l0a) at (0,5) {};
     \node (l0b) at (3,5) {};
     \node (l0c) at (3,-4) {};
     \node (l0d) at (0,-4) {};

     \draw[layer] (l0a.center)--(l0b.center) to[bend left=90] (l0c.center)--(l0d.center) to[bend left=90] (l0a.center);

     \begin{pgfonlayer}{background}
       \draw[layer] (l3a)--(l3b) to[bend left=90] (l3c)--(l3d)--(l3e) to[bend left=90] (l3a);

       \draw[layer] (l2a.center)--(l2b.center) to[bend left=90] (l2c.center)--(l2d.center)--(l2e.center) to[bend left=70] (l2a.center);

       \draw[layer] (l1a.center)--(l1b.center) to[out=0,in=90] (l1'.center) -- (l1''.center) to[out=-90,in=0] (l1c.center) -- (l1d.center) -- (l1e.center) to[bend left=70] (l1a.center);

       \draw[sepa] (l1'.center) to[bend right=40](l2a.center)--(l3a.center)--(l3e.center)--(l2e.center) to[bend right=40](l1''.center)--(l1'.center);

       \draw[sepb] (l1'.center) to[bend right=30] (l2b'.center) -- (l3a.center) -- (l3d.center) -- (l2d.center) to[bend right=20](l1''.center)--(l1'.center);

       \draw[sepc] (l1'.center)--(l2b.center) -- (l3b.center) -- (l3c.center) --(l2d.center) to[bend right=20](l1''.center)--(l1'.center);
       \draw (l1'.center) to[bend right=40](l2a.center)--(l3a.center)--(l3e.center)--(l2e.center) to[bend right=40](l1''.center)--(l1'.center);

       \draw (l1'.center) to[bend right=30] (l2b'.center) -- (l3a.center) -- (l3d.center) -- (l2d.center) to[bend right=20] (l1''.center)--(l1'.center);

       \draw (l1'.center) -- (l2b.center) -- (l3b.center) -- (l3c.center) --(l2d.center);
     \end{pgfonlayer}
   \end{tikzpicture}

  \caption{Three \wfsf{}s\iflong{} according to \cref{def-wfsf}\fi, one satisfying \cref{wfsf4a}, one satisfying \cref{wfsf4b} with $v^*=v^\dag$, and one satisfying \cref{wfsf4b} with $v^* \neq v^\dag$.  Dashed lines are the layers, solid lines are the edges, and shaded areas in each drawing are the separators, where different separators are distinguished using different shades.}
  \label{fig:seps}
\end{figure}
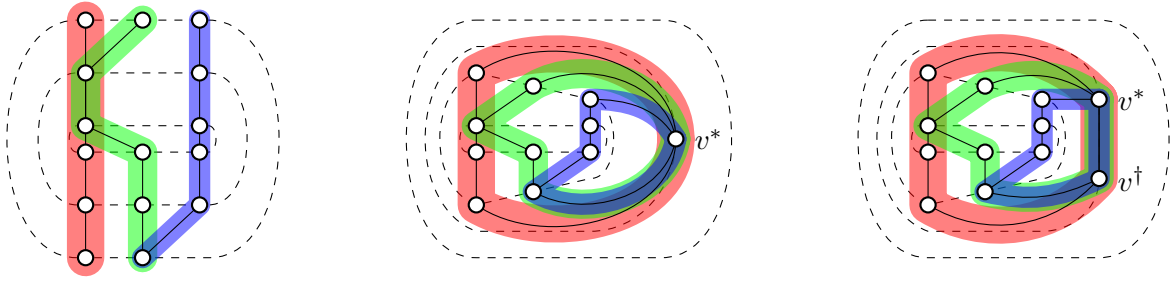

In this section, we introduce \wfsf{}s, state our main structural contribution, and compare it to results of similar nature in the literature.  Moreover, we introduce the gluing operation that \wfsf{}s are amenable to.

The separators in a \wfsf{} all have the same number of vertices and are either all induced paths or all induced cycles. Moreover, the separators stretch along consecutive layers of the $r$-outerplanar graph such that each separator contains at most two vertices from each layer and there is a one-to-one layer-correspondence between the vertices of the separators in the sequence (see \cref{fig:seps} for illustration).

\begin{definition}[\Wfsf]\label{def-wfsf}
  Let~$G=(V,E)$ be a graph with a fixed plane embedding with layers~$L_1,\dots, L_r$. A \emph{\wfsf{}} of \emph{length~$\numseps$} and \emph{width~$\sepsize$} for~$G$ is a sequence $(\A_1,S_1,\B_1), \allowbreak\dots, (\A_\numseps,S_\numseps,\B_\numseps)$ satisfying the following properties:
  \begin{compactdesc}
  \item[Linear Separation:] \leavevmode For each $i \in \{1, \ldots, \numseps\}$,
    \begin{compactenum}[(i)]
    \item\label[property]{wfsf1} $V=\A_i\cup \B_i$,
    \item\label[property]{wfsf1'} there is no edge between $\A_i\setminus \B_i$
      and $\B_i\setminus \A_i$,
    \item\label[property]{wfsf3} $\Sep_i=\A_i\cap\B_i$, $|\Sep_i|=p$, and
    \item\label[property]{wfsf2} $\A_i\subsetneq \A_{i+1}$ and $\B_i\supsetneq \B_{i+1}$.
    \end{compactenum}
  \item[Simple Shape:]\leavevmode
    \begin{compactenum}[(i)]
      \setcounter{enumi}{4}
    \item\label[property]{wfsf4} One of the following two conditions holds:
      \begin{compactenum}
      \item\label[property]{wfsf4a} for all $i \in \{1, \ldots, \numseps\}$, vertex set~$S_i$ induces a path~$(v_{i,1},\dots,v_{i,\sepsize'})$
        with $v_{i,1},v_{i,\sepsize'}\in L_1$ (in this case, $p'=p$); or
      \item\label[property]{wfsf4b} $L_1 \subseteq \A_1$ and, for all $i \in \{1, \ldots, \numseps\}$, vertex set~$S_i$ induces
        a cycle~$(\vcyca,v_{i,1},\dots,v_{i,p'},\vcycb)$, where $\vcyca$ and~$\vcycb$ are on the layer of minimum index that intersects~$S_i$ and, possibly, $\vcyca=\vcycb$ (in this case, $p'\in\{p-1,p-2\}$).
      \end{compactenum}
    \item\label[property]{wfsf4c} For $1\leq i,j\leq\numseps{}$ and $1\leq
      k,\ell \leq\sepsize'$, if~$v_{i,k}=v_{j,\ell}$, then~$k=\ell$.
    \end{compactenum}
  \item[Layering:] \leavevmode
    \begin{compactenum}[(i)]
      \setcounter{enumi}{6}
    \item\label[property]{wfsf7'} $S_i$~contains at most two vertices from each layer of $G$, and
    \item\label[property]{wfsf7} for $1\leq i,j\leq\numseps{}$ and $1\leq
      k\leq\sepsize{}'$, vertex~$v_{i,k}$ and vertex~$v_{j,k}$ are on the same layer.
    \end{compactenum}
  \end{compactdesc}
\end{definition}
\noindent Our main structural contribution in this paper is the following theorem, proved in \cref{sec:algo}.

\begin{theorem}\label[theorem]{sepseq}
  Any $\layers$-outerplanar triangulated disk with \nVer{}~vertices contains a \wfsf{} of length at least~$\lfloor\sqrt[2r]{\log(n)}/6^\layers\rfloor$ and width at most~$2\layers$.
\end{theorem}
\noindent\looseness=-1 There are several well-known approaches for constructing separators for planar/$\layers$-outerplanar graphs satisfying some of the properties of~\wfsf{}s. For example, an $\layers$-outerplanar graph $G$ has treewidth at most~$3\layers-1$ and branchwidth at most $2\layers$~\citep{Biedl15}, and thus, we can construct separator families that satisfy \crefrange{wfsf1}{wfsf3} and have width~$3\layers$ or $2\layers$ from the respective tree or branch decompositions: each bag of a tree decomposition for~$G$ is a separator for~$G$, and each edge in a branch decomposition corresponds to a separator.  Moreover, since branch decompositions are trees of bounded degree, there is an arbitrarily long path in a branch decomposition of a sufficiently large graph, and thus, an arbitrarily long sequence of separators additionally satisfying \cref{wfsf2}. However, arbitrarily large subsequences satisfying our key \crefrange{wfsf4}{wfsf7} may not be extracted from a tree/branch decomposition of $G$.  %

\looseness=-1 \emph{Layered separators}~\cite{Duj15,DMW13} yield, for  $\layers$-layer embeddings of sufficiently large graphs, arbitrarily long separator sequences of bounded width that satisfy \cref{wfsf2,wfsf7'}.  %
The ones of \citet{Duj15} yield a sequence satisfying \cref{wfsf1,wfsf3,wfsf2,wfsf7'}, but only a weaker variant of \cref{wfsf1'}, namely, that there is no edge between $(A_i\cap B_{i-1})\setminus B_i$ and $B_i\setminus (A_i\cap B_{i-1})$.  That is, each separator~$S_i$ is a separator for~$G[B_{i-1}]$ but not necessarily for~$G$. The (slightly different) ones of \citet*{DMW13} yield a sequence satisfying \cref{wfsf1,wfsf1',wfsf2,wfsf7'}, if one changes~\cref{wfsf3} so that $S_i=(A_i\cap B_i)\setminus A_{i-1}$.  That is, the separator~$A_i\cap B_i$ might use more than two vertices of a layer if these vertices are in~$A_{i-1}$.
Neither variant of the layered separators in~\cite{Duj15,DMW13} satisfies the key \cref{wfsf4,wfsf7} of \wfsf{}s. %

\paragraph{Gluing separators of a \wfsf{}.}
\label{sec:gluing}
In the following, we show a property of \wfsf{}s exploited in our algorithm for \PS. Consider the following operation for a given \wfsf:\iflong{} Pick two arbitrary separators~$(A_i,S_i,B_i)$ and~$(A_j,S_j,B_j)$ in the sequence; remove everything in the graph that is contained ``between'' the separators, that is, keep only~$A_i\cup B_j$; and glue the two separators~$S_i$ and~$S_j$ by identifying their vertices. %
\fi
\begin{definition}[Gluing]\label[definition]{def:glue} Let~$G$ be an $r$-outerplanar triangulated disk, and let~$T_i=(A_i,S_i,B_i)$
  and~$T_j=(A_j,S_j,B_j)$, $i<j$, be two separators of a \wfsf{} of width~$p$ for~$G$. We define $G({T_i\glue T_j})$ to be the graph obtained by taking the disjoint union of~$G[A_i]$ and~$G[B_j]$ and identifying each~$v_{i,k}$ in~$S_i$ with~$v_{j,k}$ from~$S_j$, for $k=1, \ldots, p$.
\end{definition}

\noindent\iflong{}As we show below,\fi{} \wfsf{}s behave nicely with respect to the gluing operation in the sense that the resulting graph is
again~$\layers$-outerplanar.
\begin{lemma}\label{lem:glue-outerp}
$G(T_i\glue T_j)$ is~$\layers$-outerplanar.
\end{lemma}

\begin{proof}First, observe that if~$S_i$ is trivial in the sense
  that~$A_i=S_i$, then the lemma holds trivially since the gluing
  operation degenerates to taking a subgraph of~$G$. By symmetry, the
  same holds if~$S_j=B_j$.
  \iflong{}%
  If~$S_i$ and~$S_j$ are nontrivial, then we distinguish two cases
  based on whether the two separators induce paths or cycles.
  \begin{case}[{$G[S_i]$ and~$G[S_j]$ are paths}]
    \else{}%
    It remains to treat nontrivial $S_i$ and~$S_j$. We show the lemma for the case where $G[S_i]$ and~$[S_j]$ are paths and defer the case where they are cycles to a full version of this article.\par
    \fi{}
    Consider the fixed embedding of~$G$. Since~$S_i$ is an induced path with its two endpoints in~$L_1$ and no other vertices in~$L_1$, it separates the region enclosed by~$L_1$ into two regions that intersect only in the vertices and edges of~$G[S_i]$.  Towards a contradiction, assume that one of these two regions contains vertices from both~$A_i\setminus S_i$ and~$B_i\setminus S_i$. Then, since $G[S_i]$~is an \emph{induced} path, there is a face in this region that contains vertices from~$A_i\setminus S_i$ and~$B_i\setminus S_i$.  Since this face is a triangle, there is an edge between~$A_i\setminus S_i$ and~$B_i\setminus S_i$; this contradicts \cref{wfsf1'}. Thus, one of the two regions contains the vertices of~$A_i\setminus S_i$ and the other one contains the vertices of~$B_i\setminus S_i$. Therefore, deleting all vertices in the region containing~$B_i\setminus S_i$ gives an embedding of~$G[A_i]$ in which all vertices of~$S_i$ and all vertices of~$A_i\cap L_1$ lie on the boundary of the outer face.
  The same statement holds for~$G[B_j]$, that is, there is an
  embedding of~$G[B_j]$ such that all vertices of~$S_j\cup (B_j\cap
  L_1)$ lie on the outer face of~$G[B_j]$. Moreover, the same is true
  for the disjoint union~$G'$ of~$G[A_i]$ and~$G[B_j]$ (using
  translation, we can assume that the embedding of~$G[B_j]$ is
  strictly to the right of~$G[A_i]$). Now, for each~$\ell$, $1\le
  \ell\le p$, add the edge~$\{v_{i,\ell},v_{j,\ell}\}$. The resulting
  graph is planar: the edge~$\{v_{i,1},v_{j,1}\}$ is between different
  connected components of~$G'$; thus it can be added without
  destroying planarity. The resulting outer face either has a
  counterclockwise face walk which contains the subsequence $(v_{j,p},
  v_{j,p-1}, \ldots, v_{j,2}, v_{j,1}, v_{i,1}, v_{i,2}, \ldots,
  v_{i,p-1}, v_{i,p})$ or this situation can be achieved by suitable
  reflections of~$G[A_i]$ or $G[B_j]$ along the horizontal axis. Now,
  adding~$\{v_{i,2},v_{j,2}\}$ replaces the old outer face by two new
  faces; one of these faces has a counterclockwise face walk with the
  subsequence~$(v_{j,p}, v_{j,p-1}, \ldots, v_{j,3}, v_{j,2}, v_{i,2},
  v_{i,3}, \ldots, v_{j,p-1}, v_{j,p})$. Hence, $(v_{i,3},v_{j,3})$ can
  be added in the same way, again creating two new faces. This process
  can be repeated until finally the edge~$\{v_{i,p},v_{j,p}\}$ is
  added. The resulting graph is planar and by contracting each of
  the~$p$ edges added to~$G'$ we obtain again a planar graph. This
  graph is exactly~$G(T_i\glue T_j)$: after these contractions, the
  neighborhood of each~$v_{i,\ell}$ is exactly the union
  of~$N(v_{i,\ell})\cap (A_i\setminus S_i)$ and~$N(v_{j,\ell})\cap
  (B_j\setminus S_j)$ plus~$v_{i,\ell-1}$ and~$v_{i,\ell+1}$ if they
  exist. All neighborhoods in~$A_i\setminus S_i$ remain the same
  in~$G$ and the constructed graph, and all neighborhoods
  in~$B_j\setminus S_j$ remain the same except that~$v_{j,\ell}$ is
  replaced by~$v_{i,\ell}$ in each neighborhood.%

It remains to show $\layers$-outerplanarity. First, observe that the
vertices of~$A_i\cap L_1$ and of~$(B_j\setminus S_j)\cap L_1$ are on
the boundary of the outer face of~$G(T_i\glue T_j)$ (if it is embedded
as described above).  This also implies that, in~$G(T_i\glue T_j)$,
each vertex of~$S_i$ is in the same layer as in~$G$.  It remains to
show that also each vertex~$v$ of~$A_i$ is in the same layer as
in~$G$. %
To this end, we exploit that $G(T_i\glue T_j)$~is a triangulated disk and, thus, that a vertex~$v$ is in~$L_i$ if and only if a shortest path from~$v$ to~$L_1$ has length exactly~$i$ \citep{Biedl15}.

\looseness=-1 Take any path witnessing that~$v\in A_i$ is in layer~$L_q$ of~$G$. If this
path contains no vertex from~$S_i$, then this path is also present
in~$G[A_i]$. If this path contains some vertex~$v_{i,k}$ from~$S_i$,
then we may assume that all the vertices that come after~$v_{i,k}$ on
this path are also in~$S_i$ (there is a direct path from $v_{i,k}$ to
the outer face in~$S_i$). Therefore, this path is present in~$G[A_i]$,
and hence in~$G(T_i\glue T_j)$, still witnessing that~$v$ is in
layer~$L_q$. By symmetry, the same holds for vertices in~$B_j$.
\iflong{}
\end{case}

\begin{case}[{$G[S_i]$~and~$G[S_j]$ are cycles}]\looseness=-1
  Assume that~$v^*\neq v^\dag$ in the following; the proof for~$v^*= v^\dag$ is completely analogous. Assume furthermore that~$A_i\setminus S_i$ and~$B_j\setminus S_j$ are nonempty; otherwise, the claim is trivially fulfilled as the gluing operation degenerates to taking a subgraph of~$G$. Let~$C_i$ and~$C_j$ denote the cycles induced by~$S_i$ and~$S_j$. Both~$C_i$ and~$C_j$ divide the plane into two regions. Since~$C_i$ is an \emph{induced} cycle, the vertices in the unbounded region for~$C_i$ can be only from~$A_i\setminus B_i$:  By \cref{wfsf4b} of~\cref{def-wfsf}, we have~$L_1\subseteq A_i$. Moreover,~$L_1\setminus S_i\neq \emptyset$ since~$L_1$ contains at least three vertices. Thus, if this region contains a vertex from~$B_i\setminus S_i$, then there is a face containing vertices of~$B_i\setminus S_i$ and of~$A_i\setminus S_i$. This face is a triangle and thus there is an edge between~$B_i\setminus S_i$ and of~$A_i\setminus S_i$. This contradicts~\cref{wfsf1'}. Thus, all vertices of~$B_i\setminus S_i$ are contained in the region enclosed by~$C_i$. By the same argument there, there can be no vertex of~$A_i\setminus S_i$ in the region enclosed by~$C_i$. When using the embedding of~$G$ for~$G[A_i]$, this implies that there is one face such that the vertex set in its boundary is exactly~$S_i$. Similarly, for~$C_j$, the unbounded region contains all vertices of~$A_j$ and no vertices of~$B_j\setminus A_j$. This implies in particular that using the embedding of~$G$ for~$G[B_j]$, the vertex set in the boundary of the outer face is exactly~$S_j$.

  Consider now the disjoint union of~$G[A_i]$ and~$G[B_j]$ where the copies of~$\vcyca$ and~$\vcycb$ introduced by adding~$G[A_i]$ and~$G[B_j]$ are denoted~$v^*_i$, $v^*_j$, and $v^\dag_i$, $v^\dag_j$, respectively. Modify the embedding of~$G[B_j]$ so that all vertices of~$B_j$ lie in the face whose boundary is~$S_i$. Assume without loss of generality that, in the combinatorial embedding of~$G[A_i]$, the face whose boundary is~$S_i$ has a counterclockwise face walk~$(v^*_i,v_{i,1},\ldots ,v^\dag_i)$ and that the outer face for~$G[B_j]$ has a counterclockwise face walk~$(v^\dag_j,v_{j,p}, \ldots , v^*_j)$. Herein, observe that the orientation of the face walk is defined by the viewpoint of the face. Thus, the order of the indices is essentially the same for both cycles. If, initially, the order for~$G[B_j]$ is the reverse, then we can use the reflection along a vertical line of the original embedding of~$G[B_j]$ instead, which reverses the order of the outer face.

  Adding the edge~$\{\vcyca_i,v^*_j\}$ to this plane graph can be done without introducing a crossing, yielding a face with face walk~$(\vcyca_i,v_{i,1},\ldots ,\vcycb_i, \vcyca_i, v^*_{j}, v^\dag_j, v_{j,p}, \ldots , v_{j,1}, v^*_j)$. Then, adding the edge~$\{v_{i,1},v_{j,1}\}$ can be again done within this embedding without introducing a crossing and such that the resulting face has the face walk~$(v_{i,1},\ldots ,\vcycb_i, \vcyca_i, v^*_{j}, v^\dag_j, v_{j,p}, \ldots , v_{j,1})$. This process can be continued, that is, we add the edge~$\{v_{i,\ell},v_{j,\ell}\}$ for increasing~$\ell$, each time obtaining a face in which the edge~$\{v_{i,\ell+1},v_{j,\ell+1}\}$ can be added without separating any vertices with higher index from the face. This is done until, finally, the edge is~$\{\vcycb_i,v^\dag_j\}$ is added. The resulting graph is planar and has an embedding with the same outer face as~$G$ in its initial embedding. Contracting each of the edges added between~$G[A_i]$ and~$G[B_j]$ gives a planar graph. After the contraction, we rename~$\vcyca_i$ to~$\vcyca$ and~$\vcycb_i$ to~$\vcycb$. By the same arguments as for the path separators, the resulting graph is exactly~$G(T_i\glue T_j)$. Thus, $G(T_i\glue T_j)$ is planar and has an embedding such that the boundary of the outer face is~$L_1$.

  \looseness=-1 It remains to show the~$\layers$-outerplanarity
  of~$G(T_i\glue T_j)$. The layering of the cycle separator and the
  fact that~$\vcyca$ and~$\vcycb$ are on the lowest layer that
  contains vertices from~$S_i$ implies that any shortest path
  from~$\vcyca$ to~$L_1$ and from~$\vcycb$ to~$L_1$ is contained
  in~$G[A_i]$ and thus in~$G(T_i\glue T_j)$. Now consider a vertex~$v$
  from~$S_i\setminus \{\vcyca,\vcycb\}$. Again the layering implies
  that there is a shortest path from~$u$ to~$L_1$ that
  contains~$\vcyca$ or~$\vcycb$ and is contained in~$G[A_i]$: Starting
  from~$v$, visit a neighbor of the current vertex that is in~$S_i$
  and in a lower layer until~$\vcyca$ or~$\vcycb$ is reached, then
  take the shortest path from this vertex to~$L_1$ (which is contained
  in~$G[A_i]$ by the previous argument). This also implies that there
  is a shortest path from any vertex~$v\in A_i\setminus S_i$ to~$L_1$
  that is completely contained in~$G[A_i]$: If a shortest path
  from~$v$ to~$L_1$ does not contain vertices from~$S_i$, then this is
  trivially true. Otherwise, a shortest path from~$v$ to~$L_1$
  contains a path from~$v$ to some vertex~$u$ in~$S_i$ that is
  completely contained in~$G[A_i]$. Then, concatenating this path with
  a shortest path from~$u$ to~$L_1$ that is completely contained
  in~$G[A_i]$ gives a shortest path from~$v$ to~$L_1$ that is
  in~$G[A_i]$. Finally, for each vertex~$v$ of~$B_j\setminus A_i$,
  observe first that, since~$L_1\subseteq A_j$, every path from~$v$
  to~$L_1$ contains a vertex of~$S_j$. Thus, let~$u$ denote the first
  vertex of~$S_j$ on a shortest path from~$v$ to~$L_1$. The shortest
  path from~$v$ to~$u$ is contained in~$G[B_j]$ and there is a
  shortest path from~$u$ to~$L_1$ which is contained in~$G[A_i]$. Both
  subpaths are contained in~$G(T_i\glue T_j)$. Thus, the distance of
  each vertex to a vertex in~$L_1$ is at least as large in~$G$ as it
  is in~$G(T_i\glue T_j)$. Hence, $G(T_i\glue T_j)$
  is~$\layers$-outerplanar. \qedhere
  \end{case}
\fi{}
\end{proof}

\section{Application: A problem kernel for Planar Support}
\label{sec:application}
\newcommand{\suppsize}{\ensuremath{2^{2^{2r(\nEd{}\cdot (r^2+r+1))}\cdot 6^{2r^2}}}}
\newcommand{\sepnum}{\ensuremath{2^{\nEd{}(\layers^2+\layers+1)}}}
We now use the existence of long \wfsf{}s to give a problem kernel for \PS{}. %
Assume that the hypergraph has an $\layers$-outerplanar support. Observe that, whenever it is convenient, we can assume that this~$\layers$-outerplanar support is a triangulated disk: triangulating interior faces and adding edges to make~$L_1$ a cycle does not increase the outerplanarity of the graph and also does not destroy the support property. Clearly, we have the desired problem kernel if $n$~can be bounded in terms of~$\nEd{}$ and~$\layers$. Otherwise, if~$\nEd{}, \layers \ll n$, then, by~\cref{sepseq}, there exists a \wfsf{} that is long in comparison with~$\nEd{}$. In this case, intuitively speaking, for at least two separators in this sequence, their ``status'' must be the same with respect to the hyperedges of~$\Hyp$ crossing them. These two separators can be glued resulting in a new graph. This new graph is not a support for~$\Hyp$ since it has less vertices. The missing vertices, however, can be ``redrawn'' to obtain an $r$-outerplanar support for~$\Hyp$\todo{ms: IFTITE there should be some more high-lvl description here.}. We formalize next the concepts discussed above.

\begin{definition}[Representative support]\label[definition]{def:representative}
  We call a graph~$G=(V,E)$ a \emph{representative support} of a hypergraph
  $\Hyp = (\Ver,\Ed)$ if
    every vertex~$u\in D:=\Ver\setminus V$ is covered by some
    vertex~$v\in V$, and
    $G$ is a support for~$\Hyp-D$.
\end{definition}
\noindent\looseness=-1 We call an~$\layers$-outerplanar support of a hypergraph~$\Hyp$ a \emph{solution}, and a representative~$\layers$-outerplanar support a \emph{representative solution} for $\Hyp$. Using~\cref{sepseq}, we now show that the size of a smallest representative solution can be upper-bounded by a function of the number~$\nEd{}$ of hyperedges of~$\Hyp$ plus the outerplanarity~$\layers$ of a solution. To this end, we first formally define the notion of two separators having the same status with respect to the hyperedges that cross the separators. To simplify the definition, we assume that, in the case of cycle separators, the vertices~$\vcyca$ and~$\vcycb$ also have indices, that is, for all~$i$, if $\vcyca = \vcycb$ then we set $\vcyca:=v_{i,p}$ and otherwise $\vcyca:=v_{i,p}$ and $\vcycb:=v_{i,p - 1}$.
\begin{definition}[Separator signature]\label[definition]{def:sepsig}\looseness=-1 Let~$(A_1,S_1,B_1), \ldots , (A_t,S_t,B_t)$ be a \wfsf{} of width~$p$ of a planar graph~$G=(V,E)$ that is a representative support for a hypergraph~$\Hyp = (\Ver,\Ed)$. The \emph{signature} of a separator~$S_i$ in this sequence is a triple~$(\Gamma_i,\phi_i,\Pi_i)$, where
  \begin{compactitem}
  \item $\Gamma_i:=\{[u]_\R\mid u\in A_i \}$ is the set of twin classes of~$A_i$,
  \item $\phi_i: \{1,\ldots,p'\}\to \{[v_{i,j}]_\R\mid u\in \Ver \},j\mapsto[v_{i,j}]_\R$
    maps each index of a vertex in~$S_i$ to
    the twin class of that vertex, and
  \item
    $\Pi_i:=\left\{(e,j,\ell)\mathrel{}\middle\vert\mathrel{}\parbox{12cm}{$e\in \Ed\wedge j<\ell\wedge v_{i,j},v_{i,\ell}\in e\wedge v_{i,j}$ and $v_{i,\ell}$ are in the same connected component of $G[B_i\cap e]$} \right\}.$
    \end{compactitem}
\end{definition}
\iflong{}\noindent Observe that, in \cref{def:sepsig}, $G$~is a representative support for~$\Hyp$, and hence, $V$~does not necessarily contain all vertices of~$\Hyp$. Moreover, the number of distinct separator signatures of a \wfsf{} is upper-bounded by a function of~$p$ and~$\nEd{}$: There are at most~$2^{\nEd{}} - 1$~twin classes in~$\Hyp$.  Furthermore, for $i<j$, we have~$A_i\subset A_j$, which implies $\Gamma_i\subseteq \Gamma_j$. Thus, either $\Gamma_i=\Gamma_{i+1}$ or $\Gamma_{i+1}$~has at least one additional twin class.  Since the number of twin classes can increase at most~$2^{\nEd{}} - 2$ times, the number of different~$\Gamma_i$ is less than~$2^{\nEd{}}$. Next, there are at most~$2^{\nEd{}}$ choices for a twin class for each~$v_{i,j}\in S_i$, leading to at most $2^{{\nEd{}p}}$ different possibilities. For the last part of the signature, we have~$\nEd{}\cdot (p^2-p)/2$ different triples, and $\Pi_i$~is an element of the power set of this set of triples. Since~$p\le 2r$, we have the following upper bound on the number of possible signatures:
\begin{observation}\label{obs:sig-num}Every \wfsf{} of a representative solution has less than \sepnum{}~different separator signatures.
\end{observation}
\fi
\begin{lemma}\label{lem:supp-size}
  If a hypergraph $\Hyp = (\Ver,\Ed)$ has a solution, then it has a
  representative solution with at most $\suppsize$ vertices.
\end{lemma}
\iflong{}
\begin{proof}
  Let~$G=(V,E)$ be a representative solution for~$\Hyp$ with the minimum number of
  vertices, and assume towards a contradiction
  that~$|V|>\suppsize$. We show that there is a representative support
  for~$\Hyp$ with less than~$|V|$ vertices. As mentioned above, we can
  assume that~$G$ is a triangulated disk.

  Since~$G$ is~$\layers$-outerplanar with more than~$\suppsize$ vertices, by~\cref{sepseq}, there is a \wfsf{} of length at least
  \begin{align*}
    \left\lfloor\frac{\sqrt[2r]{\log \suppsize}}{6^r}\right\rfloor & =
    \left\lfloor\frac{\sqrt[2r]{2^{2r(\nEd{}{\cdot (r^2+r+1)})}\cdot 6^{2r^2}}}{6^r}\right\rfloor = \left\lfloor\frac{\sepnum\cdot 6^{r}}{6^r}\right\rfloor =
    \sepnum.
  \end{align*}
 \Cref{obs:sig-num} and the
  pigeonhole principle thus imply that there are two
  separators~$T_i=(A_i,S_i,B_i)$ and~$T_j=(A_j,S_j,B_j)$, $i<j$, of
  this sequence that have the same separator signature.

  We show that the graph~$G(T_i\glue T_j)$ is a representative solution
  for~$\Hyp$. This will contradict our choice of~$G$, thus proving the
  claim. First, by \cref{lem:glue-outerp}, $G':=G(T_i\glue T_j)$ is
  an~$\layers$-outerplanar graph. Therefore, it remains to show
  that~$G'=(V',E')$ is a representative support.

  By \cref{def:glue} of the gluing operation, the vertex set of~$G'$ is~$A_i\cup (B_j\setminus S_j)$ (or equivalently, $(A_i\setminus S_i) \cup B_j$). Since the separators~$T_i$ and~$T_j$ have the same signature, we have that each twin class of~$\Hyp$ with at least one member in~$G$ has also at least one member in~$G'$: All vertices that are removed in the gluing operation are from~$A_j$ and, since~$\Gamma_i=\Gamma_j$, also in~$A_i$. Now, since each vertex of~$\Ver\setminus V$ is covered by some vertex~$v\in V$, it follows that each vertex of $\Ver\setminus V'$ is also covered by some vertex~$v'\in V'$. This shows the first of the two properties in \cref{def:representative} of representative supports. It remains to show that~$G'$ is a support for~$\Hyp[V']$.

\looseness=-1  Consider a hyperedge~$e'$ of~$\Hyp[V']$. We show that~$G'[e']$ is connected. First, let~$e$ be a hyperedge of~$\Hyp[V]$ such that~$e\cap V'=e'$, that is, $e\supseteq e'$ and the vertices of~$e$ that are not in~$e'$ are all removed during the gluing operation. Observe that such a hyperedge~$e$ exists and that, since~$G$ is a representative support of~$\Hyp$, $G[e]$~is connected. To show that~$G'[e']$ is connected we distinguish two cases.

  \begin{case}[{$e\cap S_i=\emptyset$}]
    We either have~$e\subseteq A_i\setminus S_i$ or~$e\subseteq B_j\setminus S_j$. In both cases, $G[e]=G'[e]=G'[e']$ (as~$G[A_i]=G'[A_i]$ and $G[B_j]=G'[B_j]$). Since~$G[e]$ is connected, so is $G'[e']$.
  \end{case}

  \begin{case}[{$e\cap S_i\neq\emptyset$}] %
    Observe that~$S_i\cap e$ and~$S_j\cap e$ are separators in~$G[e]$. To show that~$G'[e']$ is connected, we show three claims.
  \end{case}

  \begin{claim}[{In~$G'[e']$, each vertex~$a\in e'\cap A_i$ is connected to some vertex of~$e'\cap S_i$}]
    We have that~$G[e]$ is connected, that~$e$ contains a vertex of~$S_i$ and that~$S_i\cap e$ is a separator in~$G[e]$. Thus, $G[e]$ contains a path from~$a$ to some vertex of~$S_i$ that contains only vertices of~$A_i$. Since~$G[A_i]=G'[A_i]$ this path is also contained in~$G'$.
  \end{claim}

  \begin{claim}[{In~$G'[e']$, each vertex~$b\in e'\cap B_j$ is connected to some vertex of~$e'\cap S_i$}]
    The claim is trivially true if~$b\in S_i$. Thus, assume that~$b\in B_j\setminus S_i$. We have that~$G[e]$ is connected, that~$e$ contains a vertex of~$S_j$, and that~$S_j\cap e$ is a separator in~$G[e]$. Thus, $G[e]$ contains a path from~$b$ to some vertex~$v$ of~$S_j$ that contains only vertices of~$B_j$. Assume that this path from~$b$ to~$v$ has minimum length among all paths from~$b$ to any vertex in~$S_j$. Let~$w\in B_j$ denote the neighbor of~$v$ in this path and observe that~$w\notin S_j$. Since~$e\cap(B_j\setminus S_j)=e'\cap(B_j\setminus S_j)$ and since~$G[B_j\setminus S_j]=G'[B_j\setminus S_j]$, this path is also contained in~$G'[e'\cap(B_j\setminus S_j)]$. Now let~$v:=v_{j,k}$, that is, $v$ is the~$k$-th vertex in separator~$S_j$. By \cref{def:glue} of the gluing operation, there is in~$G'$ an edge from~$v_{i,k}$ to~$w$. Observe that~$v_{i,k}\in e'\cap S_i$ since~$\phi_i=\phi_j$ which implies that~$v_{i,k}$ and~$v_{j,k}$ are twins. Thus, $G'[e']$ contains a path from~$u$ to~$w$ to~$v_{i,k}\in e'\cap S_i$.
  \end{claim}

  \begin{claim}[{In~$G'[e']$, each pair of vertices~$u,v\in e'\cap S_i$ is connected}]
  \noindent Observe that~$u$ and~$v$ are connected by a path~$(u=p_1,\ldots,p_q=v)$ in~$G[e]$. Since~$S_i\cap e$ is a separator in~$G[e]$, this path can be decomposed into subpaths that have (respectively) only vertices in~$A_i\setminus B_i$, only vertices in~$B_i\setminus A_i$, and only vertices in~$S_i$. Let~$u=w_1,\ldots,w_x=v$ denote the vertices of this path that are in~$S_i$, that is, for each~$\ell$, $1\le \ell< x$, there is in~$G[e]$ a path from~$w_\ell$ to~$w_{\ell+1}$ that does not contain other vertices from~$S_i$. We show that, in~$G'[e']$, there is also such a path. Since each~$w_\ell\in e'$, this implies that there is a path from~$u$ to~$v$ in~$G'[e']$.

  If~$w_\ell$ and~$w_{\ell+1}$ are adjacent in~$G$, then they are also adjacent in~$G'$ and, thus, connected in~$G'[e']$. Otherwise, if the path from~$w_\ell$ to~$w_{\ell+1}$ contains vertices from~$A_i\setminus B_i$, then all these vertices are also contained in~$e'$ as~$A_i\subseteq V'$. Since~$G[A_i]=G'[A_i]$, this path is also present in~$G'[e']$. In the remaining case, the path contains vertices from~$B_i\setminus A_i$. Hence, $w_\ell$ and~$w_{\ell+1}$ are in the same connected component of~$G[B_i\cap e]$. Let~$v_{i,y}:=w_\ell$ and~$v_{i,z}:=w_{\ell+1}$. Moreover, let~$v_{j,y}$ and~$v_{j,z}$ denote the vertices that are identified with~$w_\ell$ and~$w_{\ell+1}$ in the gluing operation.

\looseness=-1  Observe that~$v_{j,y}$ and~$v_{j,z}$ are in the same connected component of~$G[B_j\cap e]$ since the separators have the same signature, which implies~$\Pi^B_i$ and~$\Pi^B_j$. Moreover, observe that~$G[B_j]$ is isomorphic to~$G'[(B_j\setminus S_j)\cup S_i]$ where the isomorphism maps each vertex of~$B_j\setminus S_j$ to itself and maps each vertex of~$S_i$ to the vertex of~$S_j$ that it is identified with. Consequently, $w_\ell$ and~$w_{\ell+1}$ are in the same connected component of~$G'[((B_j\setminus S_j)\cup S_i) \cap e')]$. Hence, there is a path from~$w_\ell$ to~$w_{\ell+1}$ in~$G'[e']$. \qedhere
  \end{claim}
\end{proof}
\fi
\noindent We now use this upper bound\iflong{} on the number of vertices in representative
solutions to obtain a problem kernel for~\PS. First, we show that
representative solutions can be extended in a particularly simple way
to obtain a solution.
\begin{lemma}\label{lem:representative-hypergraph}
  Let~$G=(V,E)$ be a representative solution for a hypergraph~$\Hyp =
  (\Ver,\Ed)$. Then, $\Hyp$ has a solution in which all vertices
  of~$\Ver\setminus V$ have degree one.
\end{lemma}
\begin{proof}
\looseness=-1  Assume, without loss of generality, that $G$~is a triangulated $\layers$-outerplanar disk.  Let $G'$~be the graph obtained from~$G$ by making each vertex~$v$ of~$\Ver\setminus V$ a degree-one neighbor of a vertex in~$V$ that covers~$v$ (such a vertex exists by the definition of representative support). Clearly the resulting graph is planar. It is also $\layers$-outerplanar: If the neighbor~$v$ of a new degree-one vertex is in~$L_1$, then $v$~can be placed in the outer face. Otherwise, $v$~can be placed in the face whose boundary contains~$v$ and a neighbor of~$v$ that lies in~$L_{i-1}$ (which exists since~$G$ is a triangulated disk~\citep{Biedl15}).

  It remains to show that~$G'$ is a support for~$\Hyp$. Consider a hyperedge~$e\in \Ed$. Since~$G$ is a representative support for~$\Hyp$, we have that~$e\cap V$ is nonempty and that $G[e\cap V]$ is connected. In~$G'$, each vertex~$u\in e\setminus V$ is adjacent to some vertex~$v\in V$ that covers~$u$. This implies, in particular, that~$v\in e$. Thus, $G'[e]$ is connected as $G'[e\cap V]$ is connected and all vertices in~$e\setminus V$ are neighbors of a vertex in~$e\cap V$.
\end{proof}

\noindent We can now use this observation to show that, if there is a twin class that is larger than a minimal representative solution, then we can \fi {}safely remove one vertex from this twin class.
\begin{lemma}\label{lem:rule}
  Let~$\Hyp = (\Ver,\Ed)$ be a hypergraph and let~$v\in \Ver$ be a
  vertex such that~$|[v]_{\R}|\ge \alpha$. If~$\Hyp$ has a representative
  solution with less than~$\alpha$ vertices, then~$\Hyp-v$ has a
  solution.
\end{lemma}\iflong
\begin{proof}
  Let~$G=(V,E)$ be a representative solution for~$\Hyp$ such that~$|V|<\alpha$. Then, at least one vertex of~$[v]_{\R}$ is not in~$V$ and we can assume, without loss of generality, that this vertex is~$v$. Thus, by~\cref{lem:representative-hypergraph}, $\Hyp$ has a support~$G'$ in which~$v$ has degree one. The graph~$G'-v$ is a support for~$\Hyp-v$: for each hyperedge~$e$ in~$\Hyp-v$, we have that~$G'[e\setminus \{v\}]$ is connected because~$v$ is not a cut-vertex in~$G'[e]$ (since it has degree one).
\end{proof}\fi
\iflong\noindent Now we combine the observations above with the fact that there are
small solutions to obtain a kernelization algorithm.\fi
\begin{theorem}
  \PS~admits a problem kernel with
  at most~$2^{\nEd{}}\cdot \suppsize$ vertices which can be computed in linear time. Hence, \PS~is
  fixed-parameter tractable with respect to~$\nEd{}+\layers$.
\end{theorem}
\iflong
\begin{proof}
  Consider an instance~$\Hyp=(\Ver,\Ed)$ of \PS{} and let~$v\in \Ver$ be contained in a twin class of size more than~$\suppsize$. By~\cref{lem:supp-size}, if~$\Hyp$ has a solution, then it has a representative solution with at most $\suppsize$ vertices. By~\cref{lem:rule}, this implies that~$\Hyp-v$ has a solution. Moreover, if~$\Hyp-v$ has a solution, then this solution is a representative solution for~$\Hyp$. By~\cref{lem:representative-hypergraph}, this implies that~$\Hyp$ has a solution. Therefore, $\Hyp$ and~$\Hyp-v$ are equivalent instances, and~$v$ can be safely removed from~$\Hyp$.

  Performing this removal can be done exhaustively in linear time~\cite{HPV99}. The removal yields an
  instance in which each twin class contains at most~$\suppsize$
  vertices; the claimed overall size bound follows since the number of
  twin classes is at most~$2^{\nEd{}}$.
\end{proof}

\begin{corollary}
\label{cor:fptwrthyperedges}
For any fixed $r \in \mathbb{N}$, the problem of deciding whether a given hypergraph $\Hyp$ has an $r$-outerplanar support is fixed-parameter tractable when 
parameterized by the number of hyperedges in $\Hyp$.
\end{corollary}
\fi

\section{Constructing \wfsf{}s}\label{sec:algo}
Throughout this section, we assume that $G$~is an~$r$-outerplanar triangulated disk on \nVer{}~vertices%
. In this section, we prove \cref{sepseq}, that is, that $G$ has a \wfsf{} of length at least \sepseqlength and width at most $2r$. The proof is by induction on the outerplanarity~$r$ and distinguishes two cases. The first case is when $G-L_1$ contains a ``large'' block~$C$. In this case, we assume by induction that~$C$ has a \wfsf{} of a certain length, which we constructively ``extend'' to a \wfsf{} of length $t$ and width at most $2r$ for~$G$ (\cref{cons:3})\iflong; we treat this case in \cref{subsec:largebicomponent}\fi.  The second case is when there is no large block in~$G-L_1$. Then either~$L_1$ is ``large'', or the number of blocks in $G - L_1$ is ``large''. We give a direct recursive construction that yields in this case a \wfsf{} of length at least \sepseqlength and of width two or three (\cref{cons-smallseps})\iflong; we treat this case in \cref{subsec:recursivecons}\fi. Observe that the second case includes the base case of outerplanar graphs. \cref{sec:proof-main} puts all together and proves \cref{sepseq}.

\subsection{\boldmath$G-L_1$ contains a large block}
\label{subsec:largebicomponent}
\looseness=-1 Let~$C$ denote a block in~$G-L_1$. We will show how a \wfsf{} for $C$ of length $t$ can be extended into a \wfsf{} of length~$t/6$ for~$G$. The resulting sequence for $G$ will be either a sequence of induced paths or a sequence of induced cycles. The following terminology will be useful when distinguishing these two possibilities. Let $P$~be an induced path in~$C$ such that the two endpoints of~$P$ lie on the outermost layer of~$C$. For a vertex~$\vcyca$ in $L_1$, we say that $\vcyca$ is a \emph{\cv{}}, or more precisely a \emph{\cv{} with respect to $P$}, if $G[V(P) \cup \{v\}]$ is an induced cycle, and for a pair of vertices~$\{\vcyca, \vcycb\}$ in~$L_1$, we say that $\{\vcyca, \vcycb\}$ is a \emph{\cp{}} (with respect to $P$) if $\vcyca$ and $\vcycb$ are adjacent, and $G[V(P) \cup \{\vcyca, \vcycb\}]$ is an induced cycle.  \iflong
\begin{lemma}
\label{lem:3cycles}
Let $C$~be a triangulated disk in~$G -L_1$.  Suppose that $C$~has a well-formed separator sequence $\Sepseq'=(\A'_1,S'_1,\B'_1),\allowbreak\dots,(\A'_\numseps,S'_\numseps,\B'_\numseps)$, where each $S'_i$, $i=1, \ldots, t$, is an induced path. Then, there are at most two distinct vertices in~$L_1$ that are cycle-vertices with respect to any path in~$\Sepseq'$.
\end{lemma}
\begin{proof}
  Proceed by contradiction. Suppose that there exist three distinct
  vertices~$v_1, v_2, v_3\in L_1$, where $v_1$~is a cycle-vertex with
  respect to~$S_{i_1}'$, $v_2$~is a cycle-vertex with respect
  to~$S_{i_2}'$, and $v_3$~is a cycle-vertex with respect
  to~$S_{i_3}'$. Let $\{u_1,w_1\}$, $\{u_2,w_2\}$, and $\{u_3,w_3\}$
  denote the vertices of~$S_{i_1}'\cap L_2$, $S_{i_2}'\cap L_2$,
  and~$S_{i_3}'\cap L_2$, respectively.  Note that each of these sets
  indeed consists of two vertices, since, by \cref{wfsf4a} of
  \wfsf{}s, each of the three induced paths starts and ends
  on the outer layer of~$C$, which is~$L_2$.  For the same reason,
  and because the outer layer of~$C$ is a cycle, there is a path~$P$
  that contains exactly one vertex from each of~$S_{i_1}'$,
  $S_{i_2}'$, and~$S_{i_3}'$, possibly some other vertices of~$L_2$,
  and only edges that are incident with the outer face of~$G-L_1$.
  Without loss of generality, assume that~$P$ contains~$u_1$, $u_2$,
  and~$u_3$. Moreover, there is also such a path~$P'$ that contains
  exactly the vertices~$w_1$, $w_2$, and~$w_3$. Consider the graph
  that is obtained from~$G$ by contracting~$P$ and~$P'$. This
  (multi)graph is a (not necessarily triangulated) $r$-outerplanar
  disk with an embedding in which~$L_1$ is the cycle incident with the
  outer face. Now, let~$u$ and~$w$ denote the vertices resulting from
  the path contractions and observe that~$u\neq w$.  Since~$u$ and~$w$
  are adjacent to~$v_1$ and~$v_2$, there is cycle~$C_u$
  containing~$u$, $v_1$, and~$v_2$ and only edges with both endpoints
  in~$L_1\cup \{u\}$. Similarly, there is a cycle~$C_w$
  containing~$w$, $v_1$, and~$v_2$ and only edges with both endpoints
  in~$L_1\cup \{w\}$. Moreover, these cycles can be chosen so that the
  regions enclosed by them intersect in~$v_1$
  and~$v_2$. The vertex~$v_3$ is contained in one of these two
  cycles. If $v_3$~is contained in~$C_u$, then it cannot be adjacent
  to~$w$ since all its edges are contained in the region enclosed
  by~$C_u$. Similarly, if $v_3$~is contained in the cycle~$C_w$, then
  it cannot be adjacent to~$u$; a contradiction.
\end{proof}
\fi
\iflong
\begin{lemma}
\label{lem:3cycles1}
Let $C$ be a triangulated disk in $G -L_1$.  Suppose that $C$ has a
well-formed separator sequence
$\Sepseq'=(\A'_1,S'_1,\B'_1),\allowbreak\dots,(\A'_\numseps,S'_\numseps,\B'_\numseps)$,
such that each~$S'_i$, $i=1, \ldots, t$, is an induced path. There can
be at most two distinct pairs~$\{v_{1}^*, v_1^{\dagger}\}$,
$\{v_{2}^*, v_2^{\dagger}\}$ of vertices such that
\begin{compactitem}
\item $\{v_{1}^*, v_1^{\dagger}$,
  $v_{2}^*, v_2^{\dagger}\}\subseteq L_1$,
\item no vertex of $\{v_{1}^*, v_1^{\dagger}$, $v_{2}^*,
  v_2^{\dagger}\}$ is a cycle-vertex with respect to any path
  in~$\Sepseq'$, and
\item $\{v_{1}^*, v_1^{\dagger}\}$,
$\{v_{2}^*, v_2^{\dagger}\}$ are \cp{}s with respect to any path
  in~$\Sepseq'$.\todo{CK: I propose to merge Lemmas 5.1 and 5.2 \quad Iyad: Please do if you have time.}{}
\end{compactitem}
\end{lemma}
\begin{proof}
  Proceed by contradiction, and assume that there exist three distinct
  pairs $\{v_{1}^*, v_1^{\dagger}\}$, $\{v_{2}^*, v_2^{\dagger}\}$,
  $\{v_{3}^*, v_3^{\dagger}\}$ in layer $L_1$ that are \cp{}s with
  respect to $S'_{i_1}$, $S'_{i_2}$, and $S'_{i_3}$, respectively.  By
  the linear separation properties of \wfsf{}s and since~$C$ is a
  triangulated disk in~$G-L_1$, there is a path~$P$ that contains
  exactly one vertex from each of~$S_{i_1}$, $S_{i_2}$, and~$S_{i_3}$,
  possibly some other vertices of~$L_2$, and only edges that are
  incident with the outer face of~$G-L_1$. Without loss of generality,
  assume that~$P$ contains~$u_1$, $u_2$, and~$u_3$. Moreover, there is
  also such a path~$P'$ that contains exactly the
  vertices~$w_1$,~$w_2$, and~$w_3$. Consider the graph~$G'$ that is
  obtained from~$G$ by contracting~$P$ and~$P'$. This (multi)graph is
  an~$r$-outerplanar disk with an embedding in which~$L_1$ is the
  cycle incident with the outer face. Now, let~$u$ and~$w$ denote the
  vertices resulting from the path contractions and observe
  that~$u\neq w$. Assume without loss of generality that~$v_{1}^*$
  and~$v_{2}^*$ are adjacent to~$u$ and~$v_{1}^\dagger$
  and~$v_{2}^\dagger$ are adjacent to~$w$. Since~$u$ is adjacent to
  $v_{1}^*$ and~$v_{2}^*$, there is cycle~$C_u$ containing~$u$,
  $v_{1}^*$ and~$v_{2}^*$ and only edges with both endpoints
  in~$L_1\cup \{u\}$ (possibly~$v_{1}^*$ and~$v_{2}^*$; in this case
  the cycle is a set of two edges). Similarly, there is a cycle~$C_w$
  containing~$w$, $v_{1}^\dagger$ and~$v_{2}^\dagger$ and only edges
  with both endpoints in~$L_1\cup \{w\}$. Moreover, these cycles can
  be chosen so that the regions enclosed by them are disjoint. The
  edge~$\{v_3^*,v_3^\dagger\}$ is contained in one of these two cycles
  since this pair is distinct from the other two \cp{}s. If
  $\{v_3^*,v_3^\dagger\}$~is contained in~$C_u$, then neither of its
  vertices can be adjacent to~$w$. Since $v_{1}^*$~and~$v_{2}^*$ are
  not cycle-vertices with respect to~$\Sepseq'$ no vertex in~$C_u$
  they are not adjacent to~$w$. Moreover, any other vertex from~$L_1$
  in~$C_u$ cannot be adjacent to~$w$ since all its incident edges are
  contained in the region enclosed
  by~$C_u$. Hence,~$\{v_3^*,v_3^\dagger\}$ can only be contained in
  the cycle~$C_w$ but then neither endpoint can be adjacent to~$u$; a
  contradiction.
\end{proof}

\noindent Using these observations, we can now describe the construction of the desired \wfsf{} of length at least~$t/6$, where~$t$ is the length of the~\wfsf{} of~$C$.  The correctness of the construction is subsequently proven in \cref{lem:cons3}.
\fi
\begin{construction} \label{cons:3} Let $G$~be an $\layers$-outerplanar triangulated disk, where $r > 1$. Suppose that $G-L_1$ has a block~$C$ such that $C$~has a \wfsf{}~$\Sepseq'=(\A_1',\Sep_1',\B_1'),\dots,\allowbreak (\A'_\numseps,\Sep'_\numseps,\B'_\numseps)$. We construct a sequence $\Sepseq=(\A_1,\Sep_1,\B_1),\dots,(\A_q,\Sep_q,\B_q)$ for~$G$ as follows: %
  \begin{case}[{$\Sepseq'$ satisfies \cref{wfsf4b}}]\label{case1}
    That is, each~$S'_i$ is a cycle.  Then, for $i=1, \ldots, \numseps$, let $\Sep_i:=\Sep_i'$, $\B_i:=\B_i'$ $\A_i:=\A_i'\cup L_1$, and define $\Sepseq:=(\A_1,\Sep_1,\B_1),\dots,(\A_\numseps,\Sep_\numseps,\B_\numseps)$.
  \end{case}

  \begin{case}[$\Sepseq'$ satisfies \cref{wfsf4a}]\label{case2}
    That is, each~$S'_i$ is a path.  Then we start by partitioning~$\Sepseq'$ into three subsequences $\Sepseq'_1$, $\Sepseq'_2$, and $\Sepseq'_3$ as follows.  For each path $S'_i \in \Sepseq'$, $i=1, \dots, \numseps$, if there exists a \cv{} with respect to $S'_i$ in $L_1$, then add $S'_i$ to $\Sepseq'_1$; otherwise, if there exists a \cp{} with respect to $S'_i$ in $L_1$, then add $S'_i$ to $\Sepseq'_2$. Finally, let $\Sepseq'_3 = \Sepseq' \setminus (\Sepseq'_1 \cup \Sepseq'_2)$. To define~$\Sepseq$, we distinguish the following cases:

    \begin{subcase}[$|\Sepseq'_1| \geq \max\{|\Sepseq'_2|, |\Sepseq'_3| \}$]\label{case2a}
      Pick a vertex $\vcyca \in L_1$ that is a \cv{} with respect to at least $|\Sepseq'_1|/2$ many paths in $\Sepseq'_1$.  Let $\langle \Sep_{i_1}, \ldots, \Sep_{i_q} \rangle$ be the sequence of cycles formed by adding~$\vcyca$ to each path in $\Sepseq'_1$ with respect to which $\vcyca$ is a \cv{}, that is, $\Sep_{i_j} = G[V(\Sep'_{i_j}) \cup \{\vcyca\}]$, where $\vcyca$ is a \cv{} with respect to $\Sep'_{i_j}$, for $j=1, \ldots, q$; the order of the cycles in the sequence is the order induced by that of the paths in $\Sepseq'_1$. Each cycle $\Sep_{i_j}$, $j=1, \ldots, q$, divides the plane into two closed regions $R_{i_j}^{1}$, $R_{i_j}^{2}$ that share $\Sep_{i_j}$ as boundary. Let~$R_{i_j}^1$ denote the region that contains $L_1$ and let~$R^{1}_{i_1}$ be a minimal region in the set of~$R_{i_j}^{1}$'s. Let $\A_{i_j}$ be the set of vertices in $R_{i_j}^{1}$, $\B_{i_j}$ be that in $R_{i_j}^{2}$, and define~$\Sepseq:= (\A_{i_1}, \Sep_{i_1}, \B_{i_1}), \ldots, (\A_{i_q}, \Sep_{i_q}, \B_{i_q})$.
    \end{subcase}

    \begin{subcase}[$|\Sepseq'_2| \geq \max\{|\Sepseq'_1|, |\Sepseq'_3| \}$]\label{case2b}
      Pick a pair of vertices $\{\vcyca,\vcycb\}$ in $L_1$ that is a \cp{} with respect to at least $|\Sepseq'_2|/2$ many paths in $\Sepseq'_2$. Let $\langle \Sep_{i_1}, \ldots, \Sep_{i_q} \rangle$ be the sequence of cycles formed by adding both $\vcyca, \vcycb$ to each path in $\Sepseq'_2$ with respect to which $\{\vcyca, \vcycb\}$ is a \cp{}, that is, $\Sep_{i_j} = G[V(\Sep'_{i_j}) \cup \{\vcyca, \vcycb\}]$, where $\{\vcyca, \vcycb\}$ is a \cp{} with respect to $\Sep'_{i_j}$, for $j=1, \ldots, q$. We define $\Sepseq$ exactly as we did in \cref{case2a} above.
    \end{subcase}

    \begin{subcase}[$|\Sepseq'_3| \geq \max\{|\Sepseq'_1|, |\Sepseq'_2| \}$]\label{case2c} Since $G$~is a triangulated disk, we can, for each vertex~$v\in L_2$, fix an arbitrary vertex~$v'\in N(v)\cap L_1$~\citep{Biedl15}.  Now, for each $S'_i=(v_{i,1},\dots,v_{i,\sepsize}) \in \Sepseq'_3$, considered with respect to its order in $\Sepseq'_3$, define $\Sep_i=(v_{i,1}',v_{i,1},\dots,v_{i,\sepsize},v_{i,\sepsize}')$.  Let $\langle \Sep_{i_1}, \ldots, \Sep_{i_q} \rangle$ be the sequence obtained in this way.  Each $\Sep_{i_j}$, $j=1, \ldots, q$, cuts $L_1$ into two cycles that enclose two closed regions~$R_{i_j}^{1}$ and~$R_{i_j}^{2}$, where $R_{i_j}^{1}$ and~$R_{i_j}^{2}$ overlap on~$\Sep_{i_j}$. For~$j=1$, let $R_{i_j}^2$~be the region that contains a vertex~$u\notin S_{i_1}$. For each~$j>1$, let $R_{i_j}^1$~be the region that contains a vertex~$u\in S_{i_1}\setminus S_{i_2}$. Now, let~$\A_{i_j}$ be the set of vertices in $R_{i_j}^{1}$, $\B_{i_j}$ be that in~$R_{i_j}^{2}$, and define $\Sepseq := (\A_{i_1}, \Sep_{i_1}, \B_{i_1}), \ldots, (\A_{i_q}, \Sep_{i_q}, \B_{i_q})$.
\qedhere
    \end{subcase}
  \end{case}
\end{construction}

\begin{lemma}
\label{lem:cons3}
The sequence $\Sepseq$ constructed in \cref{cons:3} is a well-formed separator sequence of length at least~$t/6$.\todo{rvb: The $/6$ is kind of strange, should be $/5$: there are only 5 types of separators that we can obtain: induced paths, left small ears, right small ears, left large ears, right large ears.}{}
\end{lemma}
\iflong

\begin{proof}
  \cref{cons:3} defines the sequence $\Sepseq$ based on the well-formed separator sequence~$\Sepseq':=(\A_1',\Sep_1',\B_1'),\dots,(\A'_\numseps,\Sep'_\numseps,\B'_\numseps)$ of the block $C$ in $G-L_1$ by distinguishing several cases. We show that each of these cases defines a well-formed separator sequence $\Sepseq$ for $G$.

  \begin{case}[$\Sepseq$ is constructed according to \cref{case1} of \cref{cons:3}]
    Clearly, \cref{wfsf1} is satisfied in this case because $V-L_1 = \A'_i \cup B'_i$, $i=1, \ldots, \numseps$, and we add $L_1$ to each $\A_i$. \cref{wfsf1'} can be seen as follows: By construction, $\B_i=\B'_i$ and $\Sep_i =\Sep'_i$ encloses~$B_i$. Further, since $\Sepseq'$~is a \wfsf{}, there is no edge between $\A'_i \setminus \B'_i$ and $\B'_i \setminus \A'_i$.  Since $\Sep_i$~is itself enclosed within $L_1$, there is no edge between $\A_i = \A'_i \cup L_1$ and~$B_i\setminus A_i$. \Cref{wfsf3} follows for the same reason as \cref{wfsf1'}.  \Cref{wfsf2} is trivially satisfied. \Cref{wfsf4,wfsf4c,wfsf7',wfsf7} follow because $\Sepseq'$~satisfies them and because the induced cycles in~$\Sepseq$ are the same as those in~$\Sepseq'$, which also implies that $\Sepseq$~has length~$t>t/6$.
  \end{case}

  \setcounter{case}{2}
  \begin{subcase}[$\Sepseq$ is constructed according to \cref{case2a} of \cref{cons:3}]\looseness=-1 We first prove the correctness of the construction in this case (\ie, that all the claims made in the construction are correct). The existence of a \cv{}~$\vcyca$ with respect to at least $|\Sepseq'_1|/2$~paths in~$\Sepseq'_1$ follows from:
    \begin{compactenum}[(1)]
    \item \cref{lem:3cycles}, stating that there can be at most two cycle-vertices in $L_1$ with respect to distinct paths in~$\Sepseq'$, and
    \item the definition of~$\Sepseq'_1$, which ensures that, for each path in~$\Sepseq'_1$, there exists a \cv{} in $L_1$ with respect to that path.
    \end{compactenum}
    The statement that all the cycles~$\Sep_{i_j}$, $j=1, \ldots, q$, are nested is true because $C$~is a triangulated disk and $\Sepseq'$~is a well-formed separator sequence.  Now, by the Jordan curve theorem, each cycle~$\Sep_{i_j}$, $j=1, \ldots, q$, divides the plane into two closed regions~$R_{i_j}^{1}$, $R_{i_j}^{2}$, one of which is the interior region bounded by the cycle, and the other is the exterior region.  Both regions share~$\Sep_{i_j}$ as boundary. Since each induced cycle~$\Sep_{i_j}$ consists of a path in~$C$ plus exactly one vertex in~$L_1$, which is (\ie, $L_1$) exterior to~$C$, one of the two closed regions~$R_{i_j}^{1}$ and~$R_{i_j}^{2}$ must contain~$L_1$. The nestedness of the regions follows from the nestedness of their cycle-boundaries. We now show that $\Sepseq$~satisfies all the properties of a \wfsf{}.

    \cref{wfsf1} follows trivially from the fact that each vertex is in one of the two regions. Similarly, \cref{wfsf3} follows from the fact that only the  vertices of~$S'_{i_j}$ are in~$R_{i_j}^{1}$ and $R_{i_j}^{2}$. \cref{wfsf1'} follows from the Jordan curve theorem. Now observe that the cycles $\Sep_{i_j}$, $j=1, \ldots, q$, are nested. This implies that the regions defined by these cycles form two nested sequences as well. Thus,~$R_{i_1}^{1} \subseteq \cdots \subseteq R_{i_q}^{1}$ and~$R_{i_1}^{2} \supseteq \cdots \supseteq R_{i_q}^{2}$. \cref{wfsf2} now follows from the fact that no two pairs of regions are identical. To show \cref{wfsf4b}, since $\Sepseq$ is constructed according to \cref{case2a}, each $\Sep_i \in \Sepseq$ is a cycle of the form $(\vcyca,v_{i,1},\dots,v_{i,p},\vcyca)$, where~$(v_{i,1},\dots,v_{i,p})$ is an induced path because $\Sep'_i \in \Sepseq'$. The vertex $\vcyca$ on $L_1$ is adjacent only to vertices in~$L_2$, and hence only to the two vertices $v_{i,1}, v_{i,p}$ of~$\Sep'_i$ that lie on~$L_2$. Thus, $\Sep_i$ is an induced cycle. Moreover,~$L_1 \subseteq \A_1$ follows from the definition of~$A_1$. Thus, \cref{wfsf4b} is satisfied. \Cref{wfsf4c} is satisfied since~$\Sepseq'$ satisfies it.  \Cref{wfsf7',wfsf7} follow because $\Sepseq'$ satisfies them, and each separator in $\Sepseq$ was obtained from a separator in $\Sepseq'$ by adding the same vertex~$\vcyca \in L_1$.
  \end{subcase}

  \begin{subcase}[$\Sepseq$ is constructed according to \cref{case2b} of \cref{cons:3}] The existence of a \cp{}~$\{\vcyca, \vcycb\}$ %
    with respect to at least $|\Sepseq'_2|/2$~paths in~$\Sepseq'_2$ follows from:
    \begin{compactenum}[(1)]
    \item \cref{lem:3cycles1}, stating that there can be at most two \cp{}s in $L_1$ with respect to distinct paths in~$\Sepseq'$, and
    \item the definition of $\Sepseq'_2$, which ensures that, for each path in $\Sepseq'_2$, there exists a \cp{} in $L_1$ with respect to that path.
    \end{compactenum}
    \looseness=-1 The correctness of the construction follows by similar arguments to those made in \cref{case2a} above. The proof that $\Sepseq$~satisfies the properties of a \wfsf{} is exactly the same as that for \cref{case2a} above, except when arguing that \cref{wfsf4b}~holds, that is, that each cycle~$\Sep_i$ in $\Sepseq$~is induced. Recall that~$\Sep_i$ is obtained by adding two distinct vertices~$\vcyca$ and~$\vcycb$ to an induced path~$\Sep'_i$ in~$\Sepseq'$. Since $\Sep'_i$~is an induced path, we only need to show that~$\vcyca$ and~$\vcycb$ have exactly one neighbor in~$\Sep'_i$. This is true because all the neighbors of~$\vcyca$ and~$\vcycb$ in~$C$ are in~$L_2$ and because~$\Sep'_i\in \Sepseq'_2$, which implies that the vertices in~$\Sep'_i$ are not adjacent to cycle-vertices in~$L_1$.
  \end{subcase}

  \begin{subcase}[$\Sepseq$ is constructed according to \cref{case2c} of \cref{cons:3}] To prove the correctness of the construction, first note that, for each~$S'_{i_j}=(v_{i,1},\dots,v_{i,\sepsize}) \in \Sepseq'_3$, the two vertices~$v_{i,1}'\in L_1\cap N(v_{i,1})$ and~$v_{i,\sepsize}'\in L_1\cap N(v_{i,\sepsize})$ in the extended path~$\Sep_i=(v_{i,1}',v_{i,1},\dots,v_{i,\sepsize},v_{i,\sepsize}')$ are distinct and nonadjacent: otherwise, there would be a \cv{} or a \cp{} on~$L_1$ with respect to~$S'_{i_j}$, and hence, $S'_{i_j}$~would belong to~$\Sepseq'_1$ or~$\Sepseq'_2$, not to~$\Sepseq'_3$. Thus, each~$S_{i_j}$ is a path that lies completely in the closed region of the plane delimited by~$L_1$ that contains~$C$. Therefore, each $S_i$~determines two cycles on~$L_1$ that partition the region of the plane delimited by~$L_1$ and containing~$C$ into two closed regions~$R_{i_j}^{1}$ and~$R_{i_j}^{2}$ whose boundaries overlap on~$\Sep_{i_j}$. Let~$R_{i_j}^{1}$ and~$R_{i_j}^{2}$ be the regions as specified in the construction. Then, because $\Sepseq'$~is a \wfsf{}, and by planarity, they form two nested sequences~$R_{i_1}^{1} \subseteq \cdots \subseteq R_{i_q}^{1}$ and~$R_{i_1}^{2} \subseteq \cdots \subseteq R_{i_q}^{2}$, where $R_{i_j}^{1}$~contains the prefix~$\Sep_{i_1}, \ldots, \Sep_{i_j}$ of the sequence, and $R_{i_j}^{2}$~contains the suffix~$\Sep_{i_j}, \ldots, \Sep_{i_q}$ of the sequence, for $j=1, \ldots, q$.

    Now we use these observations to prove that $\Sepseq$~satisfies the properties of a \wfsf{}. \cref{wfsf1} follows because every vertex is contained in one of the two regions. \cref{wfsf3} follows because $S_{i_j}$~is the only part shared by the two regions. \cref{wfsf1'} follows from the Jordan curve theorem and planarity. Finally, \cref{wfsf2}~follows from the nestedness of the two sequences of regions~$R_{i_j}^{1}$, $j=1, \ldots, q$, and $R_{i_j}^{2}$, $j=1, \ldots, q$ mentioned above. To show \cref{wfsf4a}, we argue that each path $\Sep_i$, obtained by adding two distinct vertices~$v_{i,1}'$ and~$v_{i,\sepsize}'$ on~$L_1$ to~$\Sep'_i=(v_{i,1},\dots,v_{i,\sepsize})$ is induced. First, observe that $\Sep'_i$~is induced. Second, as observed above, $v_{i,1}'$ and~$v_{i,\sepsize}'$ have only one neighbor in~$\Sep'_i$ (otherwise, there is a \cv{} or a cycle pair on~$L_1$ with respect to~$\Sep'_i$, contradicting the placement in~$\Sepseq'_3$). \Cref{wfsf4c} and the layering \cref{wfsf7',wfsf7} follow because $\Sepseq'$~satisfies them,  and because each~$\Sep_i \in \Sepseq$ was obtained from~$\Sep'_i \in \Sepseq'$ by adding two vertices on~$L_1$.
  \end{subcase}

\noindent\looseness=-1 It remains to give a lower bound on the length of the \wfsf{} generated by \cref{case2} of \cref{cons:3}.  Observe that one of the sequences $\Sepseq'_1$, $\Sepseq'_2$, and~$\Sepseq'_3$ generated in \cref{case2} has length at least~$t/3$ since each separator fulfills one of the three \cref{case2a,case2b,case2c}.  If this sequence is~$\Sepseq'_3$, then the constructed sequence has length at least~$t/3$. In the other two cases, the choice of the \cv{} or \cp{} guarantees that the constructed sequence has length at least~$t/6$.
\end{proof}
\fi

\subsection{\boldmath Many vertices in $L_1$ or many blocks in $G - L_1$}
\label{subsec:recursivecons}

\noindent In this case, we first generate a not necessarily well-formed sequence of \emph{separations of order two or three}, from which we later extract a sufficiently long \wfsf{}.

\begin{definition}[Separation]
  A \emph{separation} of a graph~$G$ is a pair $(A, B)$ such that%
  \begin{inparaenum}[(i)]
  \item $A \cup B = V(G)$ and
  \item there are no edges between~$A \setminus B$ and $B \setminus A$ in~$G$.
  \end{inparaenum}
  Informally, we sometimes call $A$ and $B$ the \emph{sides} of~$(A, B)$. The integer~$|A \cap B|$ is called the \emph{order} of the separation. We say that a separation~$(A, B)$ is \emph{nontrivial} if $A \setminus B \neq \emptyset \neq B \setminus A$.
\end{definition}

\noindent\looseness=-1 The construction of the sequence of separations is inductive: We start with an arbitrary trivial separation~$(A,B)$ of order two, where~$A$ is an edge incident with the outer face and $B$~contains all vertices.  With each separation~$(A,B)$, we associate a potential function~$q(B)$ that counts the number of vertices of~$L_1$ and the number of blocks remaining on the $B$-side of the separation:

\begin{definition}[Potential function]
  For a vertex set~$B$ of~$G$, let  \iflong{}\[\fi{} q(B)=|B \cap L_1| + b,\iflong{}\]\fi{} where $b$ is the number of blocks in $G[B \setminus L_1]$.
\end{definition}

\noindent Obviously, for our initial separation~$(A,B)$, the value~$q(B)$ is ``large''.  In the following, from a given separation~$(A,B)$ of order two or three such that $q(B)$~is ``large'', we construct a new separation~$(A',B')$ of order two or three such that~$q(B')\geq (q(B)-1)/\ell$ for some small value~$\ell$.  Thus, if the input graph does not have ``large'' blocks, we obtain a sequence~$\S$ of separations of order two or three whose length is roughly logarithmic in the input graph size.

\looseness=-1 The challenging part is extracting a sufficiently long \wfsf{} from~$\S$.  We will consider only separations~$(A,B)$ of order two or three such that $G[A\cap B]$ contains a path between two vertices in~$L_1$.  If $\S$~contains sufficiently many separations~$(A,B)$ such that $G[A\cap B]$ is an edge or an induced path, then we can easily extract a long \wfsf{} satisfying \cref{wfsf4a} of \cref{def-wfsf}.  However, $G[A\cap B]$ might be a triangle.  We will show that if $\S$~contains many separations that form triangles with a common edge, then these form nested cycle separators according to \cref{wfsf4b} of \cref{def-wfsf}.  Moreover, if $\S$~contains many separators forming triangles without common edges, then we will show that the ``bases'' of these triangles yield separators of a \wfsf{} of width two.

\iflong{}We now formally describe this approach and prove its correctness.  First, we\fi{} formalize the type of separations we are going to generate.\iflong{} These will be candidates for separators in the \wfsf{} we are going to create:\fi{}

\begin{definition}[Nice separation]\label{def:nicesep}
  A separation~$(A, B)$ of order two in a triangulated disk is called \emph{nice} if $A \cap B$~is an edge in~$G$ whose endpoints are both in~$L_1$. A separation~$(A,B)$ of order three is called \emph{nice} if there are exactly two vertices in~$A \cap B \cap L_1$, the vertex in $(A \cap B) \setminus L_1$ is a common neighbor of the vertices in~$A \cap B \cap L_1$, and furthermore, the three-vertex path~$P$ in~$G[A \cap B]$ with endpoints in~$L_1$ divides the region enclosed by~$L_1$ into a region~$R_A$ containing precisely the vertices in~$A$ and a region~$R_B$ containing precisely the vertices in~$B$.
\end{definition}

\iflong{}\noindent\looseness=-1 Note that the regions~$R_A$ and~$R_B$ are well-defined: $A \cap B \cap L_1$
separates the closed curve~$C$ induced by $L_1$ into two segments~$C_1, C_2$. Hereby, $P$ does not cross $C_1$ or~$C_2$ because $V(P) \setminus L_1$ lies in the region enclosed by $C$. Thus, each segment yields another closed curve when adding the path~$P$.

The following lemma shows that, in triangulated disks, all separators of size two induce exactly two separations~$(A,B)$ and $(B, A)$. This fact will be useful throughout the remainder of this section. Moreover, the lemma shows that if $A \cap B$ is an edge, it divides the region enclosed by~$L_1$ in the same way the path $P$~does for separations of order three in \cref{def:nicesep}.

\begin{lemma}\label{lem:2seps-are-edges}
  If $G$~has at least four vertices and a nontrivial separation $(A,B)$~of order two, then $G[A \setminus B]$ and $G[B \setminus A]$ are connected components in~$G-(A\cap B)$.
\end{lemma}

\begin{proof}

  Let $(A, B)$ be a nontrivial separation of order two in~$G$, let $A \cap B = \{u, v\}$, and let $A^- = A \setminus B$ and $B^-= B \setminus A$.  Note that $A^- \neq \emptyset$ and $B^- \neq \emptyset$ because $(A, B)$ is nontrivial. To prove that each of $G[A^-]$ and $G[B^-]$ is connected, it suffices to show that the number of connected components in $G - \{u, v\}$ is at most two.

  First we show that $v$ has a neighbor in $A^-$ and one in $B^-$, both different from $u$. Since $G$ is a triangulated disk, it is biconnected, whence $G-u$ is connected. Thus, there exists a simple path~$P$ in~$G-u$ from a vertex in~$A^-$ to a vertex in~$B^-$. Since $\{u, v\}$ separates $A^-$ from~$B^-$ but $u$ does not, $P$ must contain vertex $v$. Thus, indeed $v$ has the desired property.

\looseness=-1  Now suppose towards a contradiction that there are three connected components $C_1, C_2, C_3$ in $G - \{u, v\}$. At least two of these components, say $C_1$ and $C_2$, must be both in~$G[A^-]$ or both in~$G[B^-]$; assume, without loss of generality, that $C_1$ and $C_2$ are in $G[A^-]$. Again, since $G$ is a triangulated disk, it is biconnected, whence $G-u$ is connected. Thus there is a vertex~$z_1 \in C_1$ and a vertex~$z_2 \in C_2$ that are neighbors of~$v$. Also, since $G$ is a triangulated disk, there exists a path $R$ containing all the neighbors of $v$, including $u$. The graph~$R-u$ is composed of two paths $P_1$ and $P_2$. (Note that none of $z_1, z_2$ is equal to $u$.) Because $z_1$ and $z_2$ are in different connected components of $G - u$, one of them must be in $P_1$ and the other in $P_2$. Since there are no edges between $A^-$ and $B^-$, all the vertices on $P_1$ must belong to the same part as $z_1$, \ie, to~$A^-$, and all the vertices in $P_2$ must belong to the same part as $z_2$, and, hence, to~$A^-$ as well; this contradicts the fact that $v$ has a neighbor in $A^-$ and a neighbor in $B^-$, both different from~$u$.
\end{proof}
\fi
\iflong{}
\noindent By definition, nice separations correspond to paths that split the region delimited by~$L_1$ into two closed subregions.  It will often be helpful to only argue about these paths, since they, in turn, almost uniquely determine a separation:

\noindent\looseness=-1  %
\begin{definition}[Separations induced by paths]
  Let $u, w \in L_1$ such that $\{u, w\} \in E(G)$.  The edge~$\{u, w\}$ splits the closed region~$R$ of the plane delimited by~$L_1$ into two closed regions~$R_1, R_2$, whose boundaries overlap on $\{u, w\}$.  We say that $\{u, w\}$ \emph{induces} a nice separation~$(A,B)$ of order two, where one of its sides (\ie, $A$ or $B$) consists of the vertices in~$R_1$ and the other side of those in~$R_2$. Similarly, a (not necessarily induced) path~$P:=(u, v, w)$ such that $u, w \in L_1$ and $v \notin L_1$, splits~$R$ into two regions~$R_1, R_2$, whose boundaries overlap on $P$.  We say that $P$~\emph{induces} a nice separation~$(A,B)$ of order three, one of its sides consists of the vertices in~$R_1$ and the other of those in~$R_2$.
\end{definition}
\fi

\noindent As we already indicated in the beginning of this section, for nice separations~$(A,B)$ of order three, $G[A\cap B]$~might not necessarily be an induced path.  Since sequences of such separations obviously do not satisfy \cref{wfsf4a} and do not obviously satisfy \cref{wfsf4b}, it is challenging to construct \wfsf{}s from long sequences of such \emph{triangular separations}:
\begin{definition}[Triangular separation]
  A nice separation~$(A,B)$ of order three such that $G[A\cap B]$ is a triangle (\ie, a $K_3$) is called \emph{triangular} and said to \emph{induce a triangle}.  A separation $(A, B)$ is \emph{$L_1$-nontrivial} if $(A \setminus B) \cap L_1 \neq \emptyset \neq (B \setminus A) \cap L_1$.
\end{definition}

\noindent With the next lemma we show that an $L_1$-nontrivial triangular separation can be converted into a separation of order two in a unique way.  This separation of order two forms an edge; we will call it a ``base'' of the triangle.\iflong{} This is illustrated in \cref{fig:bases}.\fi{} The idea is that if we construct a sufficiently long sequence of triangular separations with mutually distinct bases, then we can construct a \wfsf{} out of the sequence of bases.

\iflong{}
\tikzstyle {lvert} = [vert, minimum size=8mm]
\tikzstyle {apart} = [pattern=north west lines,  pattern color=black, opacity=0.2]
\tikzstyle {cpart} = [pattern=north east lines, pattern color=black, opacity=0.2]
\begin{figure}
  \centering
  \begin{tikzpicture}[x=2cm,y=2cm]
    \node (l1a) at (-1,1) {};
    \node[lvert] (l1b) at (1,1) {$u$};
    \node (l1c) at (2,1) {};
    \node (l1d) at (2,0) {};
    \node[lvert] (l1e) at (1,0) {$w$};
    \node (l1f) at (-1,0) {};
    \node[lvert] (v) at (0,0.5) {$v$};

    \node (AnC) at (-0.65,0.5) {$A\cap C$};
    \node (AnC) at (0.55,0.5) {$C\setminus A$};

    \begin{pgfonlayer}{background}
      \draw[layer] (l1a.center)--(l1b.center)--(l1c.center) to[bend left=90] (l1d.center) -- (l1e.center) -- (l1f.center) to[bend left=90] (l1a.center);

      \fill[apart] (l1a.center)--(l1b.center)--(v.center)--(l1e.center) -- (l1f.center) to[bend left=90] (l1a.center);
      \fill[cpart] (l1a.center)--(l1b.center)--(l1e.center) -- (l1f.center) to[bend left=90] (l1a.center);

      \draw (l1b.center)--(l1e.center)--(v)--(l1b.center)--cycle;
      \draw[ultra thick] (l1b.center)--(l1e.center);
    \end{pgfonlayer}
  \end{tikzpicture}
  \hfill
  \begin{tikzpicture}[x=2cm,y=2cm]
    \node (l1a) at (0,1) {};
    \node[lvert] (l1b) at (1,1) {$u$};
    \node (l1c) at (3,1) {};
    \node (l1d) at (3,0) {};
    \node[lvert] (l1e) at (1,0) {$w$};
    \node (l1f) at (0,0) {};
    \node[lvert] (v) at (2,0.5) {$v$};

    \node (AnC) at (1.35,0.5) {$A\setminus C$};
    \node (AnC) at (0.5,0.5) {$A\cap C$};

    \begin{pgfonlayer}{background}
      \draw[layer] (l1a.center)--(l1b.center)--(l1c.center) to[bend left=90] (l1d.center) -- (l1e.center) -- (l1f.center) to[bend left=90] (l1a.center);

      \fill[apart] (l1a.center)--(l1b.center)--(v.center)--(l1e.center) -- (l1f.center) to[bend left=90] (l1a.center);
      \fill[cpart] (l1a.center)--(l1b.center)--(l1e.center) -- (l1f.center) to[bend left=90] (l1a.center);

      \draw (l1b.center)--(l1e.center)--(v)--(l1b.center)--cycle;
      \draw[ultra thick] (l1b.center)--(l1e.center);
    \end{pgfonlayer}
  \end{tikzpicture}

  \caption{Two triangular separations~$(A,B)$.  In both pictures, the dashed line is layer~$L_1$, part~$A$ of the triangular separations~$(A,B)$ is hatched in a north west pattern.  For each separation, the separation~$(C,D)$ of order two as in \cref{child-seps} is shown, where part~$C$ is hatched in a north east pattern.  The edge~$C\cap D$ is drawn in bold.}
  \label{fig:bases}
\end{figure}
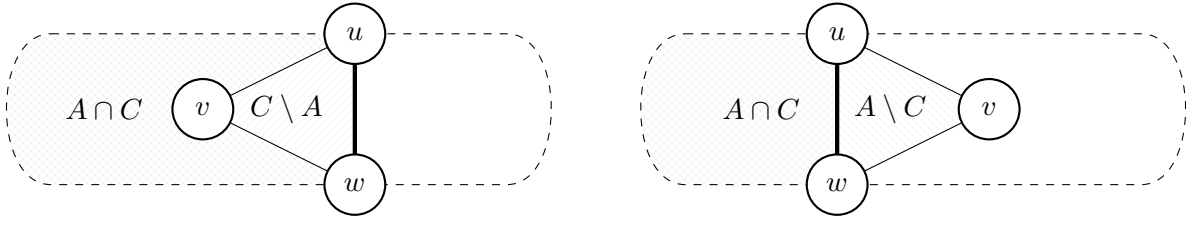
\else
\fi
\todo{Iyad: The difference between north-east and north-west pattern is not clear at all in the figure.}{}
\iflong{}Note that each nontrivial separation of order two is also $L_1$-nontrivial. But this may not be the case for separations of order three if they are triangular and one edge of the triangle is incident with the outer face. %
\fi
\begin{lemma}\label{child-seps}
  Let $(A, B)$ be a nice triangular separation in~$G$. There is a separation $(C, D)$ of order two such that $C \cap D = A \cap B \cap L_1$, and either $A \subseteq C$ and $B \cap L_1 \subseteq D$ or $B \subseteq D$ and $A \cap L_1 \subseteq C$. Moreover, if $(A, B)$ is $L_1$-nontrivial, then $(C, D)$ is unique.
\end{lemma}
\iflong
\begin{proof}
  \looseness=-1 Let $S = A \cap B \cap L_1$. Since $(A, B)$ is a triangular separation, $S \subseteq L_1$ is an edge in $G$, splitting the closed region $R$ delimited by $L_1$ into two regions $R_1$, $R_2$. Fix $R_1$ to be those region that contains the middle vertex of the path~$P$ that induces $(A, B)$ (note that not both regions can contain the middle vertex). There are two separations induced by~$S$: $(V(R_1), V(R_2))$ and $(V(R_2), V(R_1))$, where $V(R)$ denotes the set of vertices contained in region~$R$. We claim that one of these separations fulfills the conditions of the lemma.

  The three-vertex path~$P$ with endpoints in $S$ separates $R$ into a region $R_A$ containing $A$ and a region $R_B$ containing~$B$. Since $P$ cannot cross~$S$ and since the middle vertex of $P$ is in $R_1$, at least one of $R_A$ or $R_B$ is contained in $R_1$. If $R_A$ is contained in $R_1$, then we take $(C, D) := (V(R_1), V(R_2))$. Analogously, if $R_B$ is in~$R_1$, then we take $(C, D) := (V(R_2), V(R_1))$. Clearly, $(C, D)$ fulfills the condition that $A \subseteq C$ or $B \subseteq D$. To see that in the first case also $B \cap L_1 \subseteq D$, observe that the boundary of~$R_B$ differs from the boundary of~$R_2$ only in~$P$. Since $P$~intersects~$L_1$ in the same points as~$S$, the region~$R_B$ cannot enclose more vertices of~$L_1$ than $R_2$. The proof for showing that if $B \subseteq D$ then $A \cap L_1 \subseteq C$ is analogous. Hence, $(C, D)$ exists as claimed.

  It remains to show uniqueness in the case when $(A, B)$ is $L_1$-nontrivial. %
  To see this, note that ambiguity in the definition of $(C, D)$ can only occur if both~$R_A$ and~$R_B$ are contained in~$R_1$. This is impossible, however: Because $(A, B)$ is $L_1$-nontrivial, each of $A \setminus B$ and $B \setminus A$~contains a vertex of~$L_1$. One of these vertices is in~$R_1$, while the other is in~$R_2$.
\end{proof}
\fi
\begin{definition}[Base of a triangular separation]\label[definition]{bases}
  For a nice, triangular separation $(A, B)$, we call a separation~$(C, D)$ as in \cref{child-seps} a \emph{\chld{}} of $(A, B)$. %
  If, in addition, $(A, B)$ is $L_1$-nontrivial, we say that $(A, B)$~\emph{points left} if $A \subseteq C$ and that it \emph{points right} otherwise.
\end{definition}
\noindent Note that $L_1$-trivial triangular separations~$(A, B)$ have both $(V(G), A \cap B \cap L_1)$ and $(A \cap B \cap L_1, V(G))$ as \chlds{}.\iflong{} Moreover, note that the separation~$(A,B)$ shown in the left picture of \cref{fig:bases} points left, whereas the separation in the right picture points right. \fi{}
\paragraph{Inductive construction of a large sequence of nice separations.} We now show how to construct a large family of nice separations, from which a long \wfsf{} will be extracted.  That is, as described in the outline of the approach, given a separation~$(A,B)$, we want to construct a new separation~$(A',B')$ such that the potential function fulfills $q(B')\geq (q(B)-1)/\ell$ for some small number~$\ell$. The blocks play a crucial role when defining the new separation; we consider them first.\iflong{} The proof of the following lemma is illustrated in \cref{fig:gear}.\fi %

\iflong{}
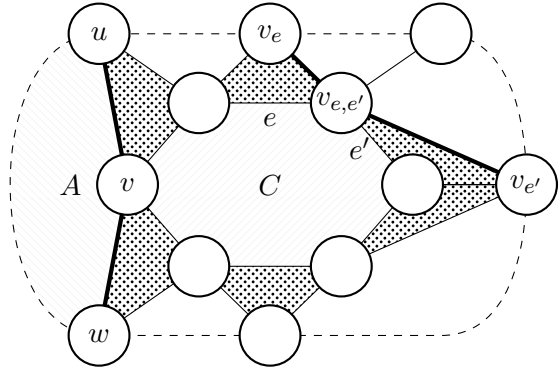
\begin{figure}
  \tikzstyle {devoid} = [pattern=crosshatch dots, opacity=1]
  \tikzstyle {block} = [cpart]
  \centering
  \begin{tikzpicture}[x=1.5cm,y=1cm]
    \node[lvert] (c1) at (180:1.25) {$v$};
    \node[lvert] (c2) at (120:1.25) {};
    \node[lvert] (c3) at (60:1.25) {$v_{e,e'}$};
    \node[lvert] (c4) at (0:1.25) {};
    \node[lvert] (c5) at (-60:1.25) {};
    \node[lvert] (c6) at (-120:1.25) {};

    \node[lvert] (l1a) at (-1.5,2) {$u$};
    \node[lvert] (l1b) at (0,2) {$v_e$};
    \node[lvert] (l1c) at (1.5,2) {};
    \node[lvert] (l1d) at (2.25,0) {$v_{e'}$};
    \node (l1e) at (1.5,-2) {};
    \node[lvert] (l1f) at (0,-2) {};
    \node[lvert] (l1g) at (-1.5,-2) {$w$};

    \node at (0,0) {$C$};
    \node at (-1.75,0) {$A$};

    \begin{pgfonlayer}{background}
      \draw (c1.center) -- (c2.center) -- (l1a.center) -- cycle;
      \fill[devoid] (c1.center) -- (c2.center) -- (l1a.center) -- cycle;

      \draw (c2.center) -- node[below] {$e$} (c3.center)  -- (l1b.center) -- cycle;
      \fill[devoid] (c2.center) -- (c3.center) -- (l1b.center) -- cycle;

      \draw (c3.center) -- (l1c.center);

      \draw (c3.center) -- node[left,pos=0.55] {$e'$} (c4.center) -- (l1d.center) -- cycle;
      \fill[devoid] (c3.center) -- (c4.center) -- (l1d.center) -- cycle;

      \draw (c4.center) -- (c5.center) -- (l1d.center) -- cycle;
      \fill[devoid] (c4.center) -- (c5.center) -- (l1d.center) -- cycle;

      \draw (c5.center) -- (c6.center) -- (l1f.center) -- cycle;
      \fill[devoid] (c5.center) -- (c6.center) -- (l1f.center) -- cycle;

      \draw (c6.center) -- (c1.center) -- (l1g.center) -- cycle;
      \fill[devoid] (c6.center) -- (c1.center) -- (l1g.center) -- cycle;

      \draw[layer] (l1a.center) -- (l1b.center) -- (l1c.center) to[out=0,in=90] (l1d.center) to[out=-90, in=0] (l1e.center) -- (l1f.center) -- (l1g.center) to[bend left=90] (l1a.center) -- cycle;;

      \fill[apart] (l1a.center) -- (c1.center) -- (l1g.center) to[bend left=90]  (l1a.center) -- cycle;

      \fill[block] (c1.center) -- (c2.center) -- (c3.center) -- (c4.center) -- (c5.center) -- (c6.center) -- cycle;

      \draw[ultra thick] (l1a.center) -- (c1.center) -- (l1g.center);

      \draw[ultra thick] (l1b.center) -- (c3.center) -- (l1d.center);
    \end{pgfonlayer}
  \end{tikzpicture}
  \caption{Illustration for the proof of \cref{lem:gearproperties1}.  The dashed line is layer~$L_1$.  A separation~$(A,B)$ is shown: the $A$-part is hatched in north west lines and the path~$A\cap B=\{u,v,w\}$ is shown in bold.  A block~$C$ of~$G-L_1$ is shown hatched in north east lines.  Each outer edge of~$C$ is incident with a dotted triangle: these triangles are devoid of vertices of~$G$.  For two  outer edges~$e,e'$ of~$C$, the path~$(v_e,v_{e,e'},v_{e'})$, which induces another separation, is shown in bold.}
  \label{fig:gear}
\end{figure}
\else
\fi

\begin{lemma}
\label{lem:gearproperties1}
Let $(A, B)$ be a nice separation of order three for $G$, where $A \cap B = \{u, v, w\}$ with $v \notin L_1$. Suppose that there is a block~$C$ in $G[B \setminus L_1]$ containing a triangle~$\{v, x_1, x_2\}$.  Then, there is a nice separation~$(A', B')$ of order three for~$G$, where $A' \cap B' =\{u', v', w'\}$, $u', w' \in L_1$, and $v' \in C$, satisfying $A \subsetneq A'$, $B \supsetneq B'$, and $q(B') \geq (q(B)-1)/|C|$.%
\end{lemma}
\iflong
\begin{proof}
  The path~$(u, v, w)$ splits the region of the plane delimited by the outermost layer~$L_1$ into two closed regions, one containing~$A$ and the other containing~$B$, whose boundaries overlap on~$u, v, w$; let $R$~be the region of the two that contains~$B$. Since $G$~is a triangulated disk, so is~$R$. Let $\gamma$~be the boundary cycle of~$R$ formed by~$u, v, w$ and one of the two paths between~$u$ and~$w$ on~$L_1$, and note that every vertex in $B \cap L_1$ is on~$\gamma$.  Since $C$~is a block in~$G[B \setminus L_1]$ and $G$~is a triangulated disk, $C$~is a triangulated disk as well.  Therefore, the outermost layer~$\gamma_C$ of~$C$ is a cycle containing~$v$. Since $C$~is a block and $R$~is triangulated, it follows from the maximality of~$C$ that, for each edge~$e$ of~$\gamma_C$, there is a vertex~$v_e \in \gamma$ such that $v_e$~forms a triangle with~$e$ (\ie, $v_e$ is adjacent to both endpoints of $e$) whose interior is devoid of vertices of~$G$. Any two consecutive edges $e$ and~$e'$ on~$\gamma_C$ such that $v_e \neq v_e'$ define a nice separation~$(A_{e,e'}, B_{e,e'})$ for~$G$ of order three.  It is induced by the (not necessarily induced) path~$(v_e, v_{e,e'}, v_e')$, where $v_{e,e'} \in \gamma_C$~is the common endpoint of the two consecutive edges~$e$ and~$e'$. (This is true because $(v_e, v_{e,e'}, v_{e'})$ is a path between two vertices on~$L_1$ that contains a vertex not in~$L_1$.) In the separation~$(A_{e,e'}, B_{e,e'})$, we designate~$A_{e,e'}$ to be the side of the separation that is delimited by the path~$(v_e, v_{e,e'}, v_e')$ and containing~$A$, and $B_{e,e'}$ to be the other side, which is contained in~$B$. Clearly, for $v_{e,e'} \neq v$, we have $A \subsetneq A_{e,e'}$ and $B_{e,e'} \subsetneq B$.  Now we go around $\gamma_C$ defining the separations~$(A_{e,e'}, B_{e,e'})$ for each two consecutive edges~$e, e'$ on~$\gamma_C$ such that $v_{e, e'} \neq v$.  The vertices~$v_e$, where $e \in \gamma_C$, belong to~$\gamma$, and every vertex in~$\gamma\setminus\{v\}$ is either equal to one of the~$v_e$'s or is situated between two of them on~$\gamma$. Therefore, every vertex in~$B \cap L_1$ belongs to~$B_{e,e'}$ for some separation~$(A_{e,e'}, B_{e,e'})$. Moreover, because $\gamma_C$~is a cycle inside the cycle~$\gamma$, it is easy to verify that each block in~$G[B\setminus L_1]$ other than~$C$ must belong to~$B_{e,e'}$ for some separation~$(A_{e,e'}, B_{e,e'})$ defined in the above process. Let $(A', B')$~be the nice separation among all the~$(A_{e,e'}, B_{e,e'})$ that maximizes the value~$q(B_{e,e'})$. From the above discussion, it follows that $A \subsetneq A'$, $B \supsetneq B'$, and $q(B') \geq (q(B)-1)/|C|$ (the minus 1 is to account for $C$).
\end{proof}
\fi

\noindent We now use \cref{lem:gearproperties1} in the inductive construction of nice separations.

\begin{lemma}\label{next-small-sep}
  \looseness=-1 Let %
  $(A, B)$~be a nice separation in~$G$ and $\ell$~be the maximum of the number~2 and the size of a largest block in $G - L_1$. If $q(B) \geq \ell$, then there is a nice separation~$(A', B')$ such that:\nopagebreak[4]
  \begin{compactenum}
  \item $A \subseteq A'$, $B \supseteq B'$;\label[ssprop]{ss:nested}
  \item $q(B') \geq (q(B)-1)/\ell$; and\label[ssprop]{ss:smallinc}
  \item if $A = A'$ or $B = B'$, then $(A, B)$ and $(A', B')$ are of different order.\label[ssprop]{ss:wfsub}
  \end{compactenum}
\end{lemma}
\iflong
\begin{proof} We distinguish between the cases of $(A,B)$~having order two or three.
  \begin{case}[$(A,B)$~is a separation of order two]
    Let $A \cap B = \{u, v\}$. Since $(A, B)$ is a nice separation, $\{u, v\}$~is an edge in~$G$. Since $q(B) \geq \ell > 2$, there is at least one vertex in~$B \setminus \{u, v\}$, and hence, there is an inner face~$F$ with $V(F) \subseteq B$ that is incident with~$\{u, v\}$ and that contains a vertex~$w \notin \{u, v\}$.
    \begin{subcase}[$w\in L_1$]
      Each of the two edges~$\{u, w\}$ and~$\{v, w\}$ is between two vertices on~$L_1$, and hence a separator for~$G$.  Thus, by \cref{lem:2seps-are-edges} $\{u, w\}$ induces a unique nice separation~$(A_1, B_1)$ of order two such that both~$u$ and~$v$ are in~$A_1$, and $\{v, w\}$ induces a unique nice separation~$(A_2, B_2)$ of order two such that both~$u$ and~$v$ are in~$A_2$. Let $(A', B')$ be the separation out of~$(A_1, B_1)$ and $(A_2, B_2)$ that maximizes~$q(B')$. Since $(A, B)$~is a nice separation of~$G$ such that $A \cap B = \{u, v\}$, and since $A'$~is the side of~$G$ that contains~$u$ and~$v$, it follows from the definition of~$A'$ that $A \subseteq A'$. Since $A' \cap B'$~is a separator contained in~$B$, it also follows that $B \supseteq B'$.  Now, the separation~$(A', B')$ was chosen to maximize~$q(B')$.  Thus, $q(B')\geq q(B)/2 \geq (q(B)-1)/\ell$ (because $q(B)$ is basically split between $q(B_1)$ and $q(B_2)$). Since $A'\ne A$ and $B'\ne B$, \cref{ss:wfsub} of the lemma is clearly satisfied by our choice of~$A'$ and~$B'$.
    \end{subcase}
    \begin{subcase}[$w\notin L_1$]
      Let $A' := A \cup \{w\}$ and $B':=B$. Since $(A, B)$~is a separation, clearly so is $(A', B')$.  Moreover, since $w \notin L_1$, $(A', B')$ is a nice separation.  Clearly, \cref{ss:nested} is fulfilled by~$(A', B')$.  \cref{ss:smallinc} is fulfilled because $B' = B$, and hence $q(B) = q(B') \geq (q(B)-1)/\ell$. Finally, $(A, B)$ and $(A', B')$ are clearly of different order, implying that \cref{ss:wfsub} holds.
    \end{subcase}
  \end{case}
  \begin{case}[$(A,B)$~is a separation of order three]
    Let $A \cap B = \{u, v, w\}$. By the definition of nice separation, there exists a vertex~$v\in (A \cap B)\setminus L_1$.  We distinguish whether $v$~has none, one, or multiple neighbors in~$B\setminus A$.
    \begin{subcase}[$v$ has no neighbors in $B \setminus A$] Since $v$~has no neighbors in~$B \setminus A$, $(A, B)$ induces an empty triangle~$\{u, v, w\}$.  Let $(A',B')$ be the unique \chld{} (see \cref{child-seps}) of~$(A, B)$  such that $A' \supseteq A$. \cref{ss:nested} holds. %
    \cref{ss:wfsub} holds because $(A, B)$~has order three and $(A', B')$~has order two. Finally, since $v \notin L_1$, we have $|B \cap L_1| = |B' \cap L_1|$. Moreover, the number of blocks in $G[B' \setminus L_1]$ is at least that in $G[B \setminus L_1]$ minus one. Therefore, $q(B') \geq q(B) - 1 \geq (q(B)-1)/\ell$, where the last inequality is true because $\ell \geq 2$.
    \end{subcase}

    \begin{subcase}[$v$ has a neighbor~$x\in(B\setminus A)\cap L_1$]\looseness=-1 It is easy to see that, in this case, each of the two paths~$(u, v, x)$ and~$(x, v, w)$ induces a nice separation of order three. Let $(A_1, B_1)$ be the separation induced by $(u, v, x)$, where $A_1$ is the side containing $w$, and $(A_2, B_2)$ that induced by $(x, v, w)$, where $A_2$ is the side containing $u$.  Note that $B_1 \subsetneq B$ (proper containment because~$w \notin B_1$) and $B_2 \subsetneq B$ (proper containment because~$u \notin B_2$). Moreover, we have $A \subsetneq A_1$ and $A \subsetneq A_2$. Let $(A', B')$~be the separation out of~$(A_1, B_1)$ and~$(A_2, B_2)$ maximizing~$q(B')$.  Since $V(B) = V(B_1) \cup V(B_2)$, it is easy to see that $q(B') \geq q(B)/2 \geq (q(B)-1)/\ell$. The above shows that \cref{ss:nested,ss:smallinc} hold.  Moreover, since the inclusions are proper, \cref{ss:wfsub} is satisfied.
    \end{subcase}

    \begin{subcase}[$v$ has exactly one neighbor~$x\in B \setminus A$, which is not in~$L_1$] Since $G$~is a triangulated disk, $x$~is a common neighbor of~$u$ and~$w$.  Moreover, the interior of the triangles~$(u, x, v)$ and~$(v, x, w)$ must be devoid of vertices of~$G$. Since $u, w \in L_1$ and $x \notin L_1$, $(u, x, w)$~induces a nice separation~$(A', B')$ of order three, where $A'=A \cup \{x\}$ and $B' = B \setminus \{v\}$.  Clearly \cref{ss:nested} is met. Moreover, since $B' = B \setminus \{v\}$, and $v \notin L_1$, $q(B') \geq q(B) -1 \geq (q(B)-1)/\ell$ (the first inequality is true because the number of blocks in~$G[B'\setminus L_1]$ is at least that in~$G[B\setminus L_1]$ minus one), and \cref{ss:smallinc} is met.  Because %
      neither $A=A'$ nor $B=B'$, \cref{ss:wfsub} holds.
    \end{subcase}

    \begin{subcase}[$v$ has at least two neighbors in $B \setminus A$ that are not in~$L_1$] Since $G$~is a triangulated disk, two neighbors~$x_1,x_2\in B\setminus A$ of~$v$ are adjacent, and hence, $v, x_1, x_2$ are part of a block~$C$ in~$G[B \setminus L_1]$. Therefore, the preconditions of \cref{lem:gearproperties1} are met, and there is a separation~$(A', B')$ satisfying \cref{ss:nested,ss:smallinc}.
    Since we have $A' \subsetneq A'$ and $B\supsetneq B'$, \cref{ss:wfsub} holds.\qedhere
  \end{subcase}
\end{case}
\end{proof}
\fi

\paragraph{Extracting a \wfsf{}.} \looseness=-1 By successively applying \cref{next-small-sep}, we can generate a long sequence of nice separations, given that our input graph is sufficiently large.  It remains to extract a long \wfsf{} from the long sequence of nice separations.
As mentioned before, we have to be careful when using nice separations~$(A,B)$ for which $G[A\cap B]$~is a triangle, since long sequences of triangles do not immediately fit into \cref{def-wfsf} of \wfsf{}s.  In \cref{child-seps}, we have already seen that $L_1$-nontrivial triangular separations can uniquely be mapped to nice separations of order two---their \chlds{}.  If the sequence of \chlds{} of triangular separations contains many mutually distinct \chld{}s, we will construct a \wfsf{} from the \chld{}s.   If not, then a long sequence of triangular separations will contain many triangles with a common base.  This is captured in the following definition\iflong{} and lemma and illustrated in \cref{fig:hinged-homogeneous}.\fi{} 
\iflong{}
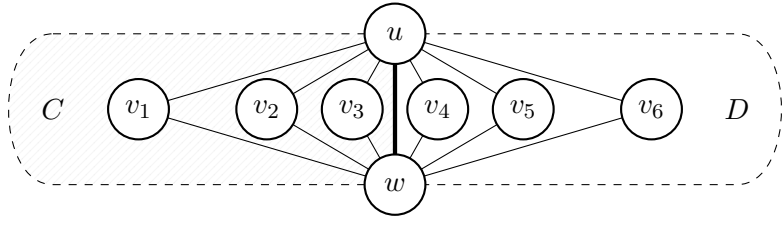
\begin{figure}
  \centering
  \begin{tikzpicture}[y=2cm, x=2.25cm]
    \node (l1a) at (-1,1) {};
    \node[lvert] (l1b) at (1,1) {$u$};
    \node (l1c) at (3,1) {};
    \node (l1d) at (3,0) {};
    \node[lvert] (l1e) at (1,0) {$w$};
    \node (l1f) at (-1,0) {};
    \node[lvert] (v1) at (-0.5,0.5) {$v_1$};
    \node[lvert] (v2) at (0.25,0.5) {$v_2$};
    \node[lvert] (v3) at (0.75,0.5) {$v_3$};

    \node[lvert] (v4) at (1.25,0.5) {$v_4$};
    \node[lvert] (v5) at (1.75,0.5) {$v_5$};
    \node[lvert] (v6) at (2.5,0.5) {$v_6$};

    \node (C) at (-1,0.5) {$C$};
    \node (D) at (3,0.5) {$D$};

    \begin{pgfonlayer}{background}
      \draw[layer] (l1a.center)--(l1b.center)--(l1c.center) to[bend left=90] (l1d.center) -- (l1e.center) -- (l1f.center) to[bend left=90] (l1a.center);

      \fill[cpart] (l1a.center)--(l1b.center)--(l1e.center) -- (l1f.center) to[bend left=90] (l1a.center);

      \draw (l1b.center)--(v1.center)--(l1e.center);
      \draw (l1b.center)--(v2.center)--(l1e.center);
      \draw (l1b.center)--(v3.center)--(l1e.center);
      \draw (l1b.center)--(v4.center)--(l1e.center);
      \draw (l1b.center)--(v5.center)--(l1e.center);
      \draw (l1b.center)--(v6.center)--(l1e.center);

      \draw[ultra thick] (l1b.center)--(l1e.center);
    \end{pgfonlayer}
  \end{tikzpicture}

  \caption{A hinged sequence of triangular separations induced by the triangles~$((u,v_i,w))_{1\leq i\leq 6}$ with their common edge drawn boldly.  All separations have a common \chld{}~$(C,D)$, whose $C$-part is hatched.  The sequence decomposes into a maximal homogeneous subsequence $((u,v_i,w))_{1\leq i\leq 3}$ pointing left and a maximal homogeneous subsequence~$((u,v_i,w))_{4\leq i\leq 6}$ pointing right.  Note that each homogeneous subsequence satisfies \cref{wfsf4b} of \wfsf{}s if we set~$\vcyca=u$ and~$\vcycb=w$.}
  \label{fig:hinged-homogeneous}
\end{figure}
\fi{}

\begin{definition}[Linear, hinged, and homogenous sequences]
  We extend the definitions of triangular, $L_1$-nontrivial, and pointing left or right, to sequences of separations in a natural way: For some property $\Pi \in \{\text{linear}, \text{hinged}, \text{homogeneous}\}$, a sequence~\S is~$\Pi$ (\ie, satisfies $\Pi$) if each separation in~\S is~$\Pi$. Moreover, a sequence~\S of separations is
  \begin{compactitem}[\emph{homogeneous}]
  \item[\emph{linear}] if, for each pair~$(A, B), (A', B')$ of consecutive separations in \S, we have $A \subsetneq A'$ and $B \supsetneq B'$;
  \item[\emph{hinged}] if it is triangular and if, for each pair $(A, B), (A', B')$ of separations in~\S, we have $A \cap B \cap L_1 = A' \cap B' \cap L_1$; and
  \item[\emph{homogeneous}] if it is hinged and either points left or points right (in particular, \S is $L_1$-nontrivial and triangular).
  \end{compactitem}
\end{definition}

\iflong\noindent When two triangular separations have the same edge between two vertices in~$L_1$, then they have a common base:

\begin{lemma}\label{linear-sep-children}
  Let $(A, B)$ and $(A', B')$ be two triangular separations such that $A \cap B \cap L_1 = A' \cap B' \cap L_1$ and $A \subsetneq A'$ and $B \supsetneq B'$. Each \chld{} of $(A, B)$ is also a \chld{} of $(A', B')$.
\end{lemma}

\begin{proof}
  Let $(C, D)$~be a \chld{} of~$(A, B)$. We prove that $(C, D)$~is a \chld{} of~$(A', B')$.  We distinguish whether $(C,D)$~is a trivial separation or not.

  \begin{case}[$(C,D)$~is a trivial separation]
    Since $(A, B)$ is triangular, $C \cap D$ is an edge in $G$. Since $C \subseteq D$ or $D \subseteq C$, the edge~$C \cap D$ is incident with the outer face. Hence, each path through~$G$ with endpoints in~$C \cap D$, containing the edge~$C \cap D$, and using otherwise only vertices not in~$L_1$, encloses a region that contains~$C \cap D$ as the only vertices in~$L_1$. Hence, by the definition of nice separation, $A' \cap L_1 \subseteq C \cap D$ or $B' \cap L_1 \subseteq C \cap D$. Since clearly both $B', A' \subseteq C \cup D$, the separation~$(C, D)$ is a \chld{} of~$(A', B')$ by definition.
  \end{case}
  \begin{case}[$(C,D)$~is a nontrivial separation]
    Clearly, $(A', B')$ has a \chld{} $(C', D')$. Moreover, $C \cap D = C' \cap D'$ since $A\cap B\cap L_1=A'\cap B'\cap L_1$.  Thus, since $(C, D)$~is nontrivial, so is $(C',D')$, which implies that $(A', B')$ is $L_1$-nontrivial.  Therefore, $(C', D')$~is unique by \cref{child-seps}.  We prove that $(C, D) = (C', D')$. Assume for the sake of contradiction that $(C, D) \neq (C', D')$. Since $C \setminus D$ and $D \setminus C$ are connected components in $G - (C \cap D)$ by \cref{lem:2seps-are-edges}, we have $C \setminus D = D' \setminus C'$ and $D \setminus C = C' \setminus D'$, meaning that $C = D'$ and $C' = D$. If both $(A, B)$ and $(A', B')$ point left, then $A \subseteq C = D'$ and $A \subseteq C'$; a contradiction since $(A,B)$~is nontrivial and, therefore, $|A| > 2$. Similarly, $(A, B)$ and $(A', B')$ cannot both point right. Assume now that $(A, B)$ points left and $(A', B')$ points right. Since $(A, B)$ points left, $A \subseteq C$ and, since it is $L_1$-nontrivial, this implies $A \cap L_1 \cap (C \setminus D) \neq \emptyset$. Thus, since $A \subsetneq A'$, one has $A' \cap L_1 \cap (D' \setminus C') \neq \emptyset$. However, since $(A', B')$~points right, one gets $A' \cap L_1 \subseteq C'$, a contradiction. The case that $(A, B)$ points right and $(A', B')'$ points left is analogous. Thus, $(C, D) = (C', D')$.\endproof
  \end{case}
\end{proof}

\noindent In particular, if $(A, B)$ and $(A', B')$ are $L_1$-nontrivial, then they share a unique \chld{}. Moreover, \cref{linear-sep-children} extends to hinged sequences.

\begin{corollary}\label{linear-hom-ss-children}
  Let \S be a linear hinged sequence of triangular separations. A \chld{} of one separation in \S is a \chld{} of each separation in \S.
\end{corollary}

\noindent\looseness=-1 Thus, we may speak of the \chld{} of a linear hinged sequence of triangular separations.

We will construct a \wfsf{} from a long sequence of nice separations as follows: if the sequence contains many triangular separations, then either we use their \chlds{} as separators if there are enough mutually distinct bases, or use a linear, hinged, homogeneous sequence as a \wfsf{} of the cycle type (satisfying \cref{wfsf4b} of \cref{def-wfsf}).  If the sequence does not contain many triangular separations, we simply throw them away.  Formally, the construction of the \wfsf{} is as follows.  \fi
\begin{construction}\label{cons-smallseps}
\looseness=-1  Let $G$ be an $\layers$-outerplanar triangulated disk and $t \in \mathbb{N}$. We construct a \wfsf{}~\T of width two or three and length~$t$ for~$G$. Let $A_1$~be any edge incident with the outer face of~$G$ and let $B_1 = V(G)$. Clearly, $(A_1, B_1)$ is a nice separation of order two. Set $i:=1$; while \cref{next-small-sep} is applicable to separation $(A_i, B_i)$, let $(A_{i + 1}, B_{i + 1})$ be the resulting (nice) separation from the application of the lemma, and set $i:=i+1$. Let $\S$~be the sequence of all the separations~$(A_i, B_i)$ defined by the above iterative process.  We distinguish the following cases:

  \begin{case}[There is a homogeneous subsequence $\S'$ of length at least~$t$ in \S]\label{sscase1} Pick a \chld{}~$(C, D)$ of a separation in~$\S'$ and define a sequence~\T as follows: If $\S'$ points left, then $\T := ((A \cup (D \setminus C), A \cap B, B \setminus (D \setminus C)))_{(A, B) \in \S'}$, inheriting the order from~$\S'$.  Otherwise, $\T := ((B \cup (C \setminus D), A \cap B, A \setminus (C \setminus D)))_{(A, B) \in \Srev'}$, where $\Srev'$~is the sequence~$\S'$ in reverse order.
  \end{case}

  \begin{case}[There is an $L_1$-trivial, hinged subsequence~$\S'$ of length at least~$t$ in~$\S$]\label{sscase4} Let $(A', B')$~be the first separation in~$\S'$ and define a sequence~\T as follows: If $L_1 \subseteq A'$ then $\T := ((A, A \cap B, B))_{(A, B) \in \S'}$, inheriting the order from~$\S'$. Otherwise, $\T := ((B, A \cap B, A))_{(A, B) \in \Srev'}$, where $\Srev'$~is the sequence~$\S'$ in reverse order.
  \end{case}

  \begin{case}[There are at least $2t$ maximal homogeneous subsequences of \S]\label{sscase2} %
    \iflong Take the sequence of their \chlds{}, inheriting the order from $\mathcal{S}$, and remove duplicates. \fi Based on the resulting sequence~$\S'$ of \chlds{}, define the sequence $\T := ((C, C \cap D, D))_{(C, D) \in S'}$ inheriting its order from~$\S'$.
  \end{case}

  \begin{case}[None of the above]\label{sscase3} Remove each triangular separation from \S. Let $\S'$ be the subsequence of \S containing only separations of order two, or only of order three, whichever is largest. Define the sequence $\T := ((A, A \cap B, B))_{(A, B) \in \S'}$ inheriting its order from $\S'$.
  \end{case}
\end{construction}
\iflong
\iflong
\noindent We next prove that the sequence \T~constructed above has length at least~$t$, regardless of the case according to which it was constructed. To this end, we have to prove that \cref{sscase2} does not discard too many duplicate \chlds{}.  We will rely on the following lemma.
\fi
\begin{lemma}\label{ss-hom-share-children}
  Let \Ssubst and $\Ssubst'$ be two maximal homogeneous subsequences of a linear, triangular sequence \Tsubst of nice separations such that each separation of~\Ssubst comes before each separation of~$\Ssubst'$ in~\Tsubst. If the \chld{} of \Ssubst is also the \chld{} of~$\Ssubst'$, then \Ssubst~points left and $\Ssubst'$~points right. Moreover, in that case, there is no separation in~\Tsubst between any pair of separations in~\Ssubst and~$\Ssubst'$.
\end{lemma}
\iflong
\begin{proof}
  Let $(A, B)$ in \Ssubst and $(A', B')$ in $\Ssubst'$ and let $(C, D)$ be their \chld{}.  We first show the lemma in the case when there is no separation between~$\Ssubst$ and~$\Ssubst'$ in~$\Tsubst$ and then show that there cannot be separations in between.

  \begin{case}[There is no separation in \Tsubst between \Ssubst and $\Ssubst'$]\label{ss-hom-share-children1} Since $\Ssubst$ and~$\Ssubst'$ are \emph{maximal homogeneous} subsequences, and there is no separation in between, they point into different directions. Assume for the sake of contradiction that \Ssubst~points right and $\Ssubst'$~points left. By \cref{bases}, that means~$B \subseteq D$.  Moreover, since $(A \cap B) \setminus L_1$ contains at least one vertex, we have $(D \setminus C) \cap A \neq \emptyset$ by \cref{child-seps}. However, from $A \subsetneq A'$ it then follows that $A' \setminus C \neq \emptyset$, which is a contradiction to $A' \subseteq C$ since $(A', B')$ points left.  Hence, if there is no separation between \Ssubst and $\Ssubst'$ in \Tsubst, then \Ssubst points left and $\Ssubst'$ points right.
  \end{case}
  \begin{case}[There is a separation~$(\hat{A}, \hat{B})$ between \Ssubst and $\Ssubst'$ in \Tsubst] We first show that $(C, D)$~is the \chld{} of each such separation $(\hat{A}, \hat{B})$.  Since concatenating~$\Ssubst$ and~$\Ssubst'$ yields a linear triangular sequence of nice separations with no separations between~$\Ssubst$ and~$\Ssubst'$, \cref{ss-hom-share-children1} shows that \Ssubst~points left and $\Ssubst'$ points right.  Without loss of generality (due to symmetry) assume that $(\hat{A}, \hat{B})$ points left. To prove that $(C, D)$ is the \chld{} of $(\hat{A}, \hat{B})$, by \cref{linear-hom-ss-children}, it suffices to prove that appending $(\hat{A}, \hat{B})$ to \Ssubst yields a homogeneous sequence, that is, $\hat{A} \cap \hat{B} \cap L_1 = A \cap B \cap L_1$. Note that $C \cap D = A \cap B \cap L_1 = A' \cap B' \cap L_1$. Since \Tsubst is linear, $A \subsetneq \hat{A}$ and $\hat{B} \supsetneq B'$, which implies $C \cap D \subseteq \hat{A} \cap \hat{B} \cap L_1$. Even equality holds since $(\hat{A}, \hat{B})$ is nice. Thus, appending $(\hat{A}, \hat{B})$ to~\Ssubst yields a homogeneous sequence, implying that $(C, D)$ is the \chld{} of $(\hat{A}, \hat{B})$ by \cref{linear-hom-ss-children}. We infer that $(C, D)$ is the \chld{} of each separation in \Tsubst between \Ssubst and $\Ssubst'$.

    \looseness=-1 Now, assume, towards a contradiction, that there are separations between~$\Ssubst$ and~$\Ssubst'$ in~$\Tsubst$.  Since \Ssubst and $\Ssubst'$ are maximal, there is a maximal triangular homogeneous subsequence $\hat{\Ssubst}$ succeeding~\Ssubst in~\Tsubst and there is a maximal triangular homogeneous subsequence $\hat{\Ssubst}'$ preceding~$\Ssubst'$ in~\Tsubst. By the choice of~$\Ssubst$ and~$\Ssubst'$, both these sequences are nonempty. By \cref{ss-hom-share-children1}, $\hat{\Ssubst}$~points right and $\hat{\Ssubst}'$~points left. However, concatenating~$\hat{\Ssubst}'$ and~$\hat{\Ssubst}$ yields a sequence that is linear, triangular, has the same \chld{} as~$\Ssubst$ and~$\Ssubst'$, and no separations between~$\hat\Ssubst'$ and~$\hat\Ssubst$.  Thus, \cref{ss-hom-share-children1} is applicable to this sequence, which leads to a contradiction since then, by \cref{ss-hom-share-children1}, $\hat{\Ssubst}$ points left and $\hat{\Ssubst}'$ points right.\qedhere
  \end{case}
\end{proof}
\fi
\noindent Furthermore, we need to prove that, after removing the triangular separations in \cref{sscase3}, there still remain sufficiently many separations. For this, we need the following lemmas.

\begin{lemma}\label{ss-dichot}
  Each linear, hinged sequence of triangular separations consists of homogeneous subsequences or is $L_1$-trivial.
\end{lemma}
\begin{proof}
  We prove that a linear, hinged sequence of triangular separations~$\S'$ that does not consist of homogeneous subsequences is $L_1$-trivial. Note that $\S'$~contains a $L_1$-trivial separation~$(A, B)$ as, otherwise, $(A, B)$ either points left or right by \cref{bases} and is thus part of a homogeneous subsequence.  Therefore, $(A, B)$ has \emph{two} trivial \chlds{}. Furthermore, by \cref{linear-hom-ss-children}, \emph{both} \chlds{} of~$(A, B)$ are \chlds{} of~$\S'$. This implies that each separation in~$\S'$ has two \chlds{} and is, by \cref{child-seps}, $L_1$-trivial. %
\end{proof}

\begin{lemma}\label{ss-no-homo}
  There are at most two maximal subsequences of the linear sequence~\S in \cref{cons-smallseps} that are both $L_1$-trivial and hinged.%
\end{lemma}
\begin{proof}
  Assume that there are three subsequences of~\S as above. Pick a separation $(A_1, B_1)$, $(A_2, B_2)$, $(A_3, B_3)$ out of each of them. A maximal hinged subsequence is consecutive in~\S, whence we may assume $A_1 \subsetneq A_2 \subsetneq A_3$ and $B_1 \supsetneq B_2 \supsetneq B_3$ without loss of generality. Furthermore, $A_1 \cap B_1 \cap L_1 \neq A_2 \cap B_2 \cap L_1 \neq A_3 \cap B_3 \cap L_1$, since \S is linear and by the maximality of the subsequences. By \cref{ss-dichot}, each of the three sequences is $L_1$-trivial. Thus, $A_i \cap B_i \cap L_1$, $i \in \{1, 2, 3\}$, is an edge incident with the outer face. Thus, there are two vertices~$u, v \in A_1 \cap B_1$, not necessarily distinct, such that $u \in A_2 \setminus B_2$ and $v \in A_3 \setminus B_3$. Let $P_2$~be the path inducing~$(A_2, B_2)$ and denote the corresponding regions by~$R^A_2, R^B_2$, which enclose~$A_2$ and~$B_2$, respectively. Analogously, let $P_3$~be the path inducing~$(A_3, B_3)$ and $R^A_3, R^B_3$ be the corresponding regions. Since $P_2$ and $P_3$~have length three, they do not cross each other. Since $R^A_2$ contains~$u$, and $R^A_3$ contains~$v$, this means that $R^A_2$ contains~$R^B_3$. Since $B_3 \supsetneq B_2$, we have $A_2 = V$. This is a contradiction to the fact that $A_2 \subsetneq A_3$.
\end{proof}

\noindent We are ready to prove a lower bound on the  length of the separator sequence~\T generated in \cref{cons-smallseps}.

\begin{lemma}\label{ss-length}
  Let %
  $\ell \geq 2$ be an upper bound on the size of each block in $G - L_1$ and $k > 0$~be a lower bound on $q(V(G))$. We can carry out \cref{cons-smallseps} in such a way that it yields a sequence of length at least \sssl.
\end{lemma}
\begin{proof}
  \looseness=-1 Let us first find a lower bound on the length~$i_m$ of the initial sequence~$\S=((A_i,B_i))_{1\leq i\leq i_m}$ that \cref{cons-smallseps} generates using \cref{next-small-sep}.  For $i<i_m$, we have $q(B_{i+1}) \geq (q(B_{i}) - 1)/\ell$ by \cref{ss:smallinc} of \cref{next-small-sep}. It is not hard to check that $i_m$~is at least the largest integer fulfilling
  \begin{align*}
    \ell &\leq \frac{q(B_1)}{\ell^{i_m - 1}} - \sum_{i = 1}^{i_m -1} \frac{1}{\ell^i},
           \intertext{which is satisfied for all $i_m$ that satisfy}
    \ell&\leq \frac{k}{\ell^{i_m - 1}} - \frac{1 - 1/\ell^{i_m}}{1 - 1/\ell} + 1.
  \end{align*}
  We claim that $i_m \geq \log_\ell k - 1$. Indeed, substituting this term for $i_m$, we obtain
  \begin{align*}
    \ell - 1 &\leq \frac{k}{\ell^{\log_\ell k - 2}} - \frac{1 - 1/\ell^{\log_\ell k - 1}}{1 - 1/\ell},\\
    \ell - 1 &\leq \ell^2 - \frac{1 - \ell/k}{1 - 1/\ell},\\
    \frac{(\ell - 1)^2}{\ell} &\leq \ell(\ell - 1) + \ell/k - 1,
  \end{align*}
  which clearly holds for all $\ell \geq 2$, $k > 0$.  We claim that carrying out \cref{cons-smallseps} with $t := \sssl$ yields a sequence~\T of length at least~$t$. Clearly, this is the case if \T~was constructed according to \cref{sscase1,sscase4}.

  Let us show that \T has length~$t$ also when it was constructed according to \cref{sscase2}. To prove this, it suffices to show that we removed at most $t$~duplicate \chlds{}.  By \cref{ss-hom-share-children}, there is no triangular separation in~\S between two sequences~$\S_1$ and~$\S_2$ with the same \chld{}.  Moreover, $\S_1$~points left and $\S_2$~points right. Thus, again by \cref{ss-hom-share-children}, both $\S_1$ and~$\S_2$ cannot share a \chld{} with any other maximal homogeneous subsequence of \S.  Hence, the duplicate \chlds{} we removed are from pairwise disjoint pairs of maximal homogeneous subsequences. Since there are $2t$ of these sequences, we removed at most $t$~duplicate \chlds{}. Hence, \T has length at least~$t$.

  Finally, consider the case that \T~was constructed according to \cref{sscase3}. To prove that \T~has length at least~$t$, it suffices to show that out of the $\log_\ell k - 1$ separations in \S, there are at most $\log_\ell k - 1 - 2t$ triangular separations. In that case, at least $2t$ separations remain in \S after removing each triangular separation, meaning that there are either at least $t$ separations of order two or at least $t$ separations of order three. Note that each triangular separation is in a hinged subsequence of $\S$. By \cref{ss-dichot}, each such subsequence is homogeneous or $L_1$-trivial. Thus, since \cref{sscase1,sscase4} did not apply when constructing \T, each hinged subsequence has length at most~$t$. Furthermore, by \cref{ss-no-homo}, there are at most two hinged subsequences that do not consist of homogeneous subsequences and, since \cref{sscase2} did not apply, there are at most $2t$~maximal homogeneous subsequences. Thus, overall, there are at most $t(2t + 2)$ triangular separations in~\S. Plugging in $t = \sssl$ we have
  \begin{align*}
    t(2t + 2) = \Bigl(\sqrt{(\log_\ell k + 1)/2} - 1\Bigr)\Bigl(2\sqrt{(\log_\ell k + 1)/2}\Bigr) &= \log_\ell k + 1 - 2\sqrt{(\log_\ell k + 1)/2} \\ & = \log_\ell k - 1 - 2t.\qedhere
  \end{align*}
\end{proof}

\paragraph{Verifying \cref{def-wfsf} of \wfsf{}s.} In the remainder of this subsection, we prove that each case of \cref{cons-smallseps} indeed yields a \wfsf, that is, we verify that the properties in \cref{def-wfsf} are satisfied.  We consider the cases in order.

\begin{lemma}\label{ss-case1}
  If \T was constructed according to \cref{sscase1} in \cref{cons-smallseps}, then \T is a \wfsf{} of width three.
\end{lemma}
\begin{proof}
\looseness=-1  Let $(\hat{A}, S, \hat{B})$ be in~\T, let $\S'$~be a sequence of separations as in \cref{sscase1} of \cref{cons-smallseps}, and let $(A, B)$~be the separation in~$\S'$ defining $(\hat{A}, S, \hat{B})$. By \cref{ss:nested} of \cref{next-small-sep}, $\S'$~is linear and thus, by \cref{linear-hom-ss-children}, $(C, D)$~as in \cref{cons-smallseps} is the \chld{} of each separation in~$\S'$.

  To verify \cref{wfsf1} of \wfsf s, it suffices to observe that, since $(A, B)$~is a separation, one has $A \cup (D \setminus C) \cup (B \setminus (D \setminus C)) = A \cup B = V(G) = B \cup A = B \cup (C \setminus D) \cup (A \setminus (C \setminus D))$. By \cref{cons-smallseps}, the set $\hat{A} \cup \hat{B}$ equals either the first or the last set in these equations.

  To verify \cref{wfsf1'}, for the sake of a contradiction, assume that there is an edge between~$\hat{A}\setminus\hat{B}$ and $\hat B\setminus\hat A$.  Consider the case that $\S'$~points left. Then, by the construction of~$\hat{A}$ and~$\hat{B}$, there is an edge between~$(A \cup (D \setminus C)) \setminus (B \setminus (D \setminus C))$ and $(B \setminus (D \setminus C)) \setminus (A \cup (D \setminus C))$. Note that the first set equals $(A \setminus B) \cup (D \setminus C)$ and the second set equals $(B \setminus A) \setminus (D \setminus C)$. Since $(A, B)$ is a separation, this implies that there is an edge between~$D \setminus C$ and~$(B \setminus A) \setminus (D \setminus C)$. Since $(D \setminus C) \cup (D \cap C) \cup (C \setminus D) = V(G)$ and $(B \setminus A) \cap (D \cap C) = \emptyset$, this implies that there is an edge between $D \setminus C$ and $C \setminus D$. This is a contradiction to the fact that $(C, D)$ is a separation. The case that $\S'$ points right is analogous.

  To verify \cref{wfsf3}, we have to show that $A\cap B=(A \cup (D \setminus C))\cap(B \setminus (D \setminus C))$ if $\S'$~points left and that $A\cap B=(B \cup (C \setminus D))\cap(A \setminus (C \setminus D)))$ if $\S'$~points right.  However, both cases are trivial since in the first case~$A \cap B \cap (D \setminus C) = \emptyset$ and in the second case $A \cap B \cap (C \setminus D) = \emptyset$.  Moreover, $|A\cap B|=3$ since $\S'$ is a triangular sequence. %

  For \cref{wfsf2}, assume that there is an element~$(\hat{A}', S', \hat{B}')$ of \T succeeding $(\hat{A}, S, \hat{B})$ and let $(A', B')$ be the separation corresponding to $(\hat{A}', S', \hat{B}')$. Consider the case that $\S'$ points left. Then $\hat{A} = A \cup (D \setminus C) \subsetneq A' \cup (D \setminus C) = \hat{A}'$ because $A \subsetneq A'$ by \cref{ss:nested,ss:wfsub} of \cref{next-small-sep} and since $A, A' \subseteq C$ because $(A, B)$ and $(A', B')$ point left. Moreover, $A, A' \subseteq C$ we have $B, B' \supseteq D \setminus C$ and hence, $B \cap (D \setminus C) = B' \cap (D \setminus C) D \setminus C$. Thus, since $B \supsetneq B'$ by \cref{ss:nested,ss:wfsub} of \cref{next-small-sep}, $\hat{B} = B \setminus (D \setminus C) \supsetneq B' \setminus (D \setminus C) = \hat{B}'$. The case that $\S'$~points right is analogous.

  We claim that \T fulfills \cref{wfsf4b}: For the second part, clearly, $S$ induces a triangle of the required form. To see the first part, assume that $(\hat{A}, S, \hat{B})$ is the first element of \T. If $(A, B)$ points left, then $B \cap L_1 \subseteq D$ by \cref{child-seps}. Since $L_1 \subseteq A \cup B \cap L_1$ we thus have $L_1 \subseteq A \cup (D \setminus C) = \hat{A}$, as required. The case that $(A, B)$ points right is analogous.

  Finally, \cref{wfsf4c,wfsf7',wfsf7} directly follow from the fact that $(A, B)$ is nice and of order three \todo{Iyad: Does it? For (viii) (and (vi)), don't you need that all the separations in ${\cal T}$ are of order 3?\\ms: I don't see why it doesn't. Both properties should be fulfilled for nice separations.}.
\end{proof}

\begin{lemma}\label{ss-case4}
  If \T is constructed according to \cref{sscase4} in \cref{cons-smallseps}, then \T is a \wfsf{} of width three.
\end{lemma}
\begin{proof}
  \cref{wfsf1,wfsf1',wfsf3} are fulfilled since $\S'$ is a sequence of triangular separations. The linearity of~$\S'$ implies \cref{wfsf2}.

  We claim that \T fulfills \cref{wfsf4b}: Clearly, the intersections~$A \cap B$ in the definition of~\T induce triangles of the required form.  It remains to show $L_1 \subseteq A'$ for the first separation~$(A', B')$ in $\S'$.  Assume that this is not the case.  Then, since $\S'$~is $L_1$-trivial, we have $L_1 \subseteq B'$. Furthermore, since $\S'$~is hinged, the path induced by each separation in~$\S'$ touches~$L_1$ in the same place. Thus, $L_1 \subseteq B$ for each separation~$(A, B)$ in~$\S'$. Since \T~contains $(B, A \cap B, A)$ in this case, it satisfies \cref{wfsf4b}.

  Finally, \cref{wfsf4c,wfsf7',wfsf7} directly follow from the fact that each separation in~$\S'$ is nice.
\end{proof}

\noindent For \cref{sscase2}, we first need to show that all the considered \chlds{} are nontrivial and that their induced separators differ.

\begin{lemma}\label{ss-children-nontrivial}
  Each \chld{} is nontrivial in the sequence $\S'$ of \chlds{} in \cref{sscase2} of \cref{cons-smallseps} and, moreover, for each pair of \chlds{} $(C, D)$ and~$(C', D')$ in~$\S'$, we have $C \cap D \neq C' \cap D'$.
\end{lemma}

\begin{proof}
  We first prove that, for each pair of \chlds{} $(C, D)$ and~$(C', D')$ in~$\S'$, we have $C \cap D \neq C' \cap D'$.  Observe that $(C, D)$ and~$(C', D')$ are the \chlds{} of two different maximal homogeneous subsequences~\T and~$\T'$ of~$\S'$.  Towards a contradiction, assume that $C \cap D = C' \cap D'$.  Then, \T and~$\T'$ point to different directions.  Without loss of generality, each separation of~\T comes before each separation of~$\T'$ in~$\S'$. Let $(A, B)$ and $(A', B')$~be separations in~\T and~$\T'$, respectively. Then, $(A, B)$ and $(A', B')$ fulfill the preconditions of \cref{linear-sep-children}. This implies that $(C, D) = (C', D')$ by \cref{child-seps} since~\T and~$\T'$ are homogeneous, and therefore $L_1$-nontrivial by definition.  We now have our contradiction, since the sequence~$\S'$ output by \cref{sscase2} does not contain duplicates.

  It remains to show that all \chlds{} in~$\S'$ are nontrivial.  Let $(A,B)$~be a separation whose \chld{}~$(C, D)$ is in~$\S'$. Since $(A, B)$ is part of a homogeneous sequence, it either points left or right and, in particular, is $L_1$-nontrivial. Without loss of generality, assume that $(A, B)$~points left, the other case is similar. By \cref{bases}, $A \subseteq C$. Moreover $(A \cap B) \setminus D \neq \emptyset$, whence we have $C \setminus D \neq \emptyset$. By $L_1$-nontriviality $(B \setminus A) \cap L_1 \neq \emptyset$. Since, by \cref{child-seps}, $B \cap L_1 \subseteq D$, we also have $D \setminus C \neq \emptyset$. Hence, $(C, D)$ is nontrivial.
\end{proof}

\begin{lemma}\label{ss-case2}
  If \T was constructed according to \cref{sscase2} in \cref{cons-smallseps}, then \T is a \wfsf{} of width two.
\end{lemma}

\begin{proof}
  \looseness=-1 Clearly, as each \chld{} of a separation is itself a separation, \cref{wfsf1,wfsf1'} of \wfsf s are fulfilled.  It is easy to see that \cref{wfsf3} is fulfilled as well.

  To prove \cref{wfsf2} first recall that each \chld{} in~$\S'$ is nontrivial by \cref{ss-children-nontrivial}. Let $(C, D)$ be the \chld{} in~$\S'$ of some separation~$(A, B)$ in \S that is not the last one and let $(C', D')$~be the \chld{} in~$\S'$ belonging to a separation $(A', B')$ with a higher index than $(A, B)$ in~\S. We claim that $C \subsetneq C'$ and $D \supsetneq D'$. Since, by definition of nice separations (\cref{def:nicesep}), $D$~is uniquely determined once~$C$ and~$C \cap D$ are defined, whence it suffices to prove that~$C \subsetneq C'$. Since both $C \cap D$ and $C' \cap D'$~are edges in~$G$ with endpoints in~$L_1$, they subdivide the region enclosed by~$L_1$ into three regions. One of these regions, $R$, is incident with~$C \cap D$ and not incident with~$C' \cap D'$ because $C \cap D \neq C' \cap D'$ by \cref{ss-children-nontrivial}. Again, since $C \cap D \neq C' \cap D'$, there is a vertex~$v \in (C \cap D) \setminus (C' \cap D')$. Moreover, since $(C, D)$~is nontrivial, $v$~has a neighbor~$u \in L_1 \setminus D$ contained in~$R$.  We distinguish two cases: vertex~$u$ is contained in~$A$ or~$B$.

  \begin{case}[$u \in A$] Then, as $A \subsetneq A'$ and~$u \in L_1$, we have $u \in C'$ by \cref{child-seps}. Furthermore, $u \in C \setminus D$ by the choice of~$u$. Since $C \setminus D$~is connected (\cref{lem:2seps-are-edges}) and $(C \setminus D) \cap C' \cap D' = \emptyset$, we obtain $C \setminus D \subseteq C'$. Since $C \cap D$ is connected to~$u$ via~$v$ and~$v \notin C \cap D \cap C' \cap D'$, furthermore $C \cap D \subseteq C'$ holds. Hence, if $u \in A$, we have~$C \subsetneq C'$.
  \end{case}

  \begin{case}[$u\in B$]
    We lead this case to a contradiction.  By \cref{def:nicesep} of nice separations, $B$~is enclosed by the curve induced by the vertices in~$L_1$ that are also in~$R$ and a path~$P$ contained in~$A \cap B$.  Since $|(V(P) \cap L_1) \setminus \{v\}| \leq 1$\todo{ms: $v \in L_1$ indeed.

\par\quad rvb: Argh, in all other proofs, $v$ is used as a name for a vertex NOT in $L_1$}{}, we have that $|B \cap (C' \cap D')| \leq 1$. This is a contradiction, since $B' \subsetneq B$ and there are two vertices in $C' \cap D' = A' \cap B' \cap L_1$.
  \end{case}

\noindent  This proves \cref{wfsf2}.  Finally, since for each separation~$(C, D)$ in~\T, we have that $C \cap D$~is an edge and $C \cap D \subseteq L_1$, \cref{wfsf4c,wfsf7',wfsf7} are fulfilled.%
\end{proof}

\begin{lemma}\label{ss-case3}
  If \T was constructed according to \cref{sscase3} in \cref{cons-smallseps}, then \T is a \wfsf{} of width two or three.
\end{lemma}
\begin{proof}
  Clearly, each object in \T is a separation, hence \cref{wfsf1,wfsf1'} are fulfilled. By the choice of the separations in~\T, also \cref{wfsf2} holds: for each two consecutive separations~$(A, B)$ and~$(A', B')$ in~\S, the initial sequence in \cref{cons-smallseps}, we have $A \subsetneq A'$ and $B \supsetneq B'$ by  \cref{ss:wfsub} of \cref{next-small-sep}. Hence, the same holds for the subsequence~\T.  \cref{wfsf4a} is fulfilled since none of the separations~$(A, B)$ in \T induces a triangle, and since $G[A \cap B]$ contains a path whose endpoints are in~$L_1$ by \cref{def:nicesep} of nice separation. Finally, also \cref{wfsf4c,wfsf7',wfsf7} follow directly from \cref{def:nicesep}.
\end{proof}

\noindent Combining \cref{ss-length,ss-case1,ss-case4,ss-case2,ss-case3}, we obtain the following corollary. Note that, in the case that $G$~is outerplanar, there are no nice separations of order three in~$G$, and hence, the \wfsf{} constructed in \cref{cons-smallseps} has width two.
\begin{corollary}\label{small-seps}
  Let $G$ be an $\layers$-outerplanar triangulated disk, let $\ell \geq 2$ be an upper bound on the size of each block in $G - L_1$ and $k > 0$ a lower bound on $q(V(G))$. Using \cref{cons-smallseps}, we can construct a \wfsf{} of width at most three and of length at least~\sssl. If $G$ is outerplanar, then the width of the sequences is two instead.
\end{corollary}
\fi

\subsection{Proof of \cref{sepseq}}\label{sec:proof-main}
Let $G$~be an $\layers$-outerplanar triangulated disk with \nVer{}~vertices.  If $t:=\lfloor\sqrt[2r]{\log(n)}/6^\layers\rfloor\leq 0$, then the theorem follows trivially.  Thus, in the following, assume that $t\geq 1$, that is, $\log(n)\geq 6^{2r^2}\geq9$.  We use induction on the outerplanarity~$\layers$ of~$G$ in order to prove that we can construct a \wfsf{} of width at most~$2r$ and length at least~$t$ for~$G$.

\allowdisplaybreaks For $\layers=1$, note that $q(V) = |L_1| = n$ and there are no blocks in $G - L_1$. Hence, by \cref{small-seps}, we can construct a \wfsf{} of width two and length at least $\sqrt{(\log(n) + 1)/2} - 1$.  In this case, $\sqrt{(\log(n) + 1)/2} - 1 \geq t$ is implied by
\begin{align*}
  \sqrt{\log(n)}/2 - 1 &\geq \sqrt{\log(n)}/6, &
  \sqrt{\log(n)} - \sqrt{\log(n)}/3 & \geq 2,\text{ and} &
  \log(n) & \geq 9\text{.}
\end{align*}
Now, assume that the statement is true for $(r-1)$-outerplanar triangulated disks, where $r-1\geq 1$, and we prove it for $r$-outerplanar triangulated disks.  Assume that there is a block~$C$ in~$G - L_1$ with at least $s := 2^{(\log(n))^{\frac{r-1}{r}}}$~vertices.  It is not hard to see that, since $G$~is a triangulated disk and since~$C$ contains at least three vertices, $C$~is a triangulated disk. Moreover, $C$ has outerplanarity at most~$r - 1$.  Therefore, we can apply the inductive hypothesis to~$C$.  We thus infer that there is a \wfsf{} \S for $C$ of width at most~$2(r - 1)$ and length at least~$\lfloor(\log(s))^{1/2(r - 1)}/6^{r-1}\rfloor = \lfloor6\cdot \sqrt[2r]{\log(n)}/6^\layers\rfloor \geq 6t$. By \cref{lem:cons3}, we can extend \S to a \wfsf{} of width at most~$2\layers$ and length at least~$t$ for~$G$. %

Now, assume that each block in~$G - L_1$ contains at most~$s$ vertices. Note that $q(V(G)) \geq |L_1| + (n - |L_1|)/s$, and hence, $q(V) \geq n/s$. By \cref{small-seps}, there is a \wfsf{} of width at most~$3 \leq 2\layers$ and length at least~$\sqrt{(\log_{s}(q(V)) + 1)/2} - 1$.\iflong{} We claim that this sequence has length at least~$t$. This claim follows from the following list of inequalities that are pairwise equivalent:
\begin{align*}
  \sqrt{(\log_{s}(n/s) + 1)/2} - 1 &\geq t,\\
  \log_{s}(n/s) &\geq 2(t + 1)^2 -1,\\
  \frac{\log(n/s)}{\log(s)} &\geq 2(t + 1)^2-1,\\
  \frac{\log(n)}{\log(s)}  &\geq 2(t + 1)^2, \\
  \frac{\log(n)}{(\log(n))^{\frac{r-1}{r}}} &\geq 2(t + 1)^2,\\
  (\log(n))^{\frac{1}{r}} &\geq 2(t + 1)^2.
\end{align*}
\noindent Now the last inequality is true because, for $r, t \geq 1$, we have $2(t + 1)^2  \leq 8t^2 \leq \sqrt[r]{\log(n)}$.\qed
\fi{}

\iflong{}
\section{Supplement I: Beware of removing twins}\label{sec:supp}
As mentioned in \cref{sec:intro}, earlier results about \PS~(\eg, \citet[p. 179]{makinen}, \citet[p. 346]{bkmsv}, \citet[p. 399]{KKS08}) have been obtained under the assumption that the input hypergraph is \emph{twinless}.
\begin{figure}[t]\centering
  \begin{tikzpicture}[>=stealth,rounded corners=2pt,x=35pt,y=25pt,
    vert/.style={circle,thick,draw,inner sep=0pt,minimum size=2mm},
    ivert/.style={rectangle,thick,fill=black,draw,inner
      sep=0pt,minimum size=2mm},
    tvert/.style={circle,thick,fill=black,draw,inner
      sep=0pt,minimum size=2mm}]
    \node (a) at (0,0) [vert,label=left:$a$] {};
    \node (b) at (3,4.5) [vert,label=above:$b$] {};
    \node (c) at (6,0) [vert,label=right:$c$] {};
    \node (d) at (3,1) [ivert,label=below:$d$] {};

    \draw [-,thick] (a)-- (d) -- (b);
    \draw [-,thick] (a) to [bend left=45] (b);
    \draw [-,thick] (b) to [bend left=45] (c);
    \draw [-,thick] (c) -- (d);
    \draw [-,thick] (a)-- (c);

    \node (vab) at (0.9,1) [vert,label=right:$v_a$] {};
    \node (uab) at (1.75,2) [ivert,label=right:$v_d$] {};
    \node (vba) at (2.25,3) [vert,label=right:$v_b$] {};

    \draw [-,thick] (a)-- (vab) -- (uab) -- (vba) -- (b);
    \draw [-,thick] (d) -- (uab);

    \node (vbc) at (3.75,3) [vert,label=left:$u_b$] {};
    \node (ubc) at (4.25,2) [ivert,label=left:$u_d$] {};
    \node (vcb) at (5.1,1) [vert,label=left:$u_c$] {};

    \draw [-,thick] (b)-- (vbc) -- (ubc) -- (vcb) -- (c);
    \draw [-,thick] (d) -- (ubc);

    \node (t) at (1.25,2.5) [tvert,label=left:$t$] {};
    \node (tp) at (4.75,2.5) [tvert,label=right:$t'$] {};

    \draw [dotted,thick] (a)-- (t);
    \draw [dotted,thick] (b)-- (t);
    \draw [dotted,thick] (vab)-- (t);
    \draw [dotted,thick] (vba)-- (t);

    \draw [dotted,thick] (c)-- (tp);
    \draw [dotted,,thick] (b)-- (tp);
    \draw [dotted,,thick] (vbc)-- (tp);
    \draw [dotted,,thick] (vcb)-- (tp);

  \end{tikzpicture}

  \caption{An example showing that twins can be essential for obtaining a ($2$-outer)planar support. The hyperedges consist of the solid lines in the figure, plus $\{a,v_a,t,t',c\},\allowbreak \{a,v_b,t,t',c\},\allowbreak
  \{b,v_a,t,t',c\},\allowbreak \{b,v_b,t,t',c\},  \{b,u_b,t,t',a\}$, $\{b,u_c,t,t',a\}, \{c,u_b,t,t',a\}, \{c,u_c,t,t',a\}$. The vertices~$t$ and $t'$ are twins, and $\Hyp$~has a (2-outer)planar support but $\Hyp-t$ does not.}
\label{fig:counter}
\end{figure}
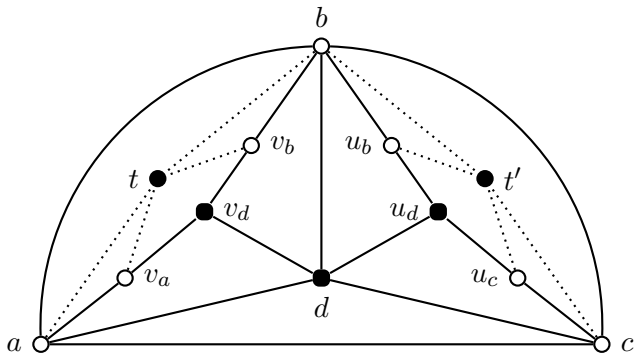
In \cref{fig:counter}, we provide a concrete example why removing twins is harmful.  The vertex-set of the hypergraph~$\Hyp$ shown in~\cref{fig:counter} is:
\[\Ver:=\{a,b,c,d,v_a,v_b,v_d,u_b,u_c,u_d,t,t'\}.\]
We construct~$\Hyp$ in such a way that~$t$ and~$t'$ are twins
and~$\Hyp$ has a planar support but~$\Hyp-t$ does not. Let~$\Ed$ contain the size-two hyperedges for each \emph{solid} edge shown in
\cref{fig:counter}. Observe that the embedding for this graph, and
thus for any support for a hypergraph containing these edges, is
basically fixed: The
set~$\{a,b,c,d\}$ induces a~$K_4$ and any plane embedding of the~$K_4$
has one face for each triangle. Now the path from~$a$ to~$b$
containing~$v_d$ has to be inside the face that is incident with~$a$, $b$,
and~$c$ as~$v_d$ is a neighbor of~$d$. The same holds for the path
from~$b$ to~$c$ containing~$u_d$. The remaining hyperedges contained in~$\Ed$ are:
\begin{align*}
  \{a,v_a,t,t',c\}, &&\{a,v_b,t,t',c\},&&
  \{b,v_a,t,t',c\}, &&\{b,v_b,t,t',c\},\\
  \{b,u_b,t,t',a\}, &&\{b,u_c,t,t',a\},&&
  \{c,u_b,t,t',a\}, &&\{c,u_c,t,t',a\}.
\end{align*}
Adding~$t$ and~$t'$ and the dotted edges to the solid graph gives a
planar support for~$\Hyp$.

Now consider the hypergraph~$\Hyp-t$. The solid edges are still
hyperedges of this hypergraph, hence the embedding of the solid edges
and their incident vertices is fixed in any support. Now observe that
in any planar support either~$v_a$ is not adjacent to~$b$ or~$v_b$ is
not adjacent to~$a$. Moreover, neither of these vertices can be
adjacent to~$c$. Thus, to make the graph induced by the hyperedges
containing~$v_a$ or~$v_b$ connected, $t'$ must be adjacent to one of
the two vertices in any support. For the same reason, $t'$ must be
adjacent to~$u_b$ or~$u_c$. This is not possible since each face is
either incident with~$v_a$ and/or~$v_b$ or with~$u_b$ and/or~$u_c$ but
not both. Hence, $\Hyp-t$ has no planar support. Therefore, removing
one vertex of a twin class can transform a yes-instance into a
no-instance.

The above example can be generalized to make the twin classes arbitrarily
large: Copy the vertex set above~$\ell$ times, and let~$$V_i := \{a_i,
  b_i, c_i, d_i, v_{i,a} , v_{i,b}, v_{i,d} , u_{i,b} , u_{i,c} ,
  u_{i,d}, t_i, t_i'\}$$ denote the vertex set of the~$i$-th
copy. Within each copy, add the size-two hyperedges as in the example
above. Then, further add a distinguished vertex~$v^*$, and add the
size-two hyperedges~$\{a_i,v^*\}$, $\{b_i,v^*\}$, and~$\{c_i,v^*\}$ to
the hypergraph. After this, for any support~$G$, $G[V_i]$ is
constrained to be a copy of~\cref{fig:counter}, and for each
copy, $v^*$ is adjacent to the three vertices~$a_i$, $b_i$, and~$c_i$.
Now, let~$A:=\{a_i\mid 1\le i \le\ell\}$, $B:=\{b_i\mid 1\le i
\le\ell\}$, and~$C:=\{c_i\mid 1\le i \le\ell\}$, $V_a:=\{v_{i,a}\mid 1\le i \le\ell\}$, $V_b:=\{v_{i,b}\mid 1\le i \le\ell\}$, $U_b:=\{u_{i,b}\mid 1\le i \le\ell\}$, $U_c:=\{u_{i,c}\mid 1\le i \le\ell\}$, and let~$T:=\{t_i\mid
1\le i \le\ell\}\cup \{t'_i\mid 1\le i \le\ell\}$. Then, add the
hyperedges
\begin{align*}
  A\cup C\cup V_a\cup T\cup \{v^*\}, &&A\cup C\cup V_b\cup T\cup \{v^*\},\\
  B\cup C\cup V_a\cup T\cup \{v^*\}, &&B\cup C\cup V_b\cup T\cup \{v^*\},\\
  B\cup A\cup U_b\cup T\cup \{v^*\}, &&B\cup A\cup U_c\cup T\cup \{v^*\},\\
  C\cup A\cup U_b\cup T\cup \{v^*\}, &&C\cup A\cup U_c\cup T\cup \{v^*\}.
\end{align*}
The instance is a yes-instance as~$v^*$ can be used to ``connect''
partial solutions for each~$V_i$ that are obtained by copying the
solution for the simple example. Moreover, each face that is initially
incident with~$\{a_i,b_i,v_{i,a},v_{i,b},v_{i,d}\}$ has to contain at
least one vertex of~$T$. Since there are~$\ell$ such faces, removing
one vertex of~$T$ transforms the yes-instance into a no-instance.

\section{Supplement II: Non-uniform fixed-parameter tractability}\label{sec:nonuni-fpt}

\begin{theorem}\label{nonuni-fpt}
  Let $\Pi$~be a graph property that is closed under adding degree-one vertices.  There is a function~$f\colon\mathbb N\to\mathbb N$ such that, for each fixed~$m \in \mathbb{N}$, there is an algorithm that determines whether a given hypergraph $\Hyp$ with $\nEd{}$ hyperedges has a support satisfying~$\Pi$ in time $f(\nEd{})\cdot\poly(|\Hyp|)$.
\end{theorem}

\noindent Note that the theorem holds in particular for $\Pi$ being planarity or $r$-outerplanarity.

\newcommand{\good}{$\Pi$-supportable}
\begin{proof}[Proof sketch]
  Let us call a hypergraph \emph{\good} if it admits a $\Pi$-support. We define a quasi-order~$\preceq$ on the family of hypergraphs with $\nEd{}$~hyperedges such that, if $\Hyp$~is \good{} and $\Hyp \preceq \G$, then $\G$~is \good.  We show that, for every~$m \in \mathbb{N}$, the family~$\Psi_m$ of \good{} hypergraphs that are minimal under~$\preceq$ is finite.

  \looseness=-1 To define~$\preceq$, we say that $\Hyp \preceq \G$ if $\Hyp$~can be obtained from~$\G$ by iteratively removing a vertex that has a twin.  If we allow zero removals so that $\preceq$~is reflexive, it is clear that $\preceq$~is a quasi-order. Furthermore, if $\Hyp$~has a $\Pi$-support~$G$, then adding the missing twins of a vertex~$v$ in~$\G$ as degree-one vertices to~$v$ in~$G$ will yield a $\Pi$-support for~$\G$. Thus indeed, if $\Hyp$~is \good, so is~$\G$.

  \looseness=-1 To see that $\Psi_m$~is finite, consider the representation of an $\nEd{}$-hyperedge hypergraph~$\Hyp$ as a $2^m$-tuple~$t_\Hyp$, each entry of which represents the size of a distinct twin class. The set of such tuples is quasi-ordered by the natural extension of~$\leq$ as $(a_1, \ldots, a_\ell) \leq (b_1, \ldots, b_\ell)$ if and only if $a_i \leq b_i$ for each $i \in \{1, \ldots, \ell\}$. Moreover, \citet[Theorem 2.3]{Hig52} has shown that every infinite sequence of $2^m$-tuples contains two tuples $t_1, t_2$ with $t_1 \leq t_2$. Assume that $\Psi_m$ is infinite; then there is an infinite subset $\Psi_m'$ of hypergraphs which have the same (nonempty) twin classes. For hypergraphs~$\Hyp, \G$ with the same twin classes, $t_\Hyp \leq t_\G$ implies $\Hyp \preceq \G$. Thus, $\Psi_m'$ implies an infinite sequence of tuples that are pairwise incomparable under $\leq$, a contradiction. Hence, $\Psi_m$ is finite.

  Finally, to obtain an algorithm for every fixed~$\nEd{}$ as in the theorem, we hard-wire the family~$\Psi_m$ of \good{} hypergraphs minimal with respect to~$\preceq$ into the algorithm. The algorithm simply checks whether its input hypergraph~$\Hyp$ fulfills $\mathcal{F} \preceq \Hyp$ for some $\mathcal{F} \in \Psi_m$, which clearly can be checked in polynomial time for each~$\mathcal{F}\in\Psi_m$.
\end{proof}
\else
\fi

\section{Conclusion}

So far, we only used \wfsf{}s for kernelization. It is interesting to find more algorithmic applications of these separators, for example in a divide and conquer algorithm for \PS.  We would also like to point out that \wfsf{}s can be used to find %
nicely structured separators in $r$-outerplanar graphs that are \emph{not} triangulated disks: an $r$-outerplanar graph~$G$ can be turned into a triangulated disk~$G'$ such that each vertex remains on its layer~\citep{Biedl15}.  Hence, by computing a long \wfsf{} for~$G'$, one obtains  for~$G$ a separator sequence satisfying \cref{wfsf1,wfsf1',wfsf3,wfsf2,wfsf4c,wfsf7',wfsf7} of \cref{def-wfsf}. Additionally, the graph~$G[S_i]$ is a subgraph of an induced path or a cycle.  Using this approach, we conjecture that it is also possible to apply our arguments to the variant of \PS{} that asks for a planar support with a minimum number of edges.

\paragraph{Acknowledgments.}
We thank anonymous referees for very helpful comments improving the presentation of the results and for pointing out\iflong{}
\cref{nonuni-fpt}.
\fi

\iflong{}
René van Bevern, Iyad Kanj, and Manuel Sorge acknowledge support by the DFG, project DAPA (NI~369/12).  Parts of this work were done while René van Bevern was employed at TU Berlin and during a six month stay of Iyad Kanj at TU Berlin.
\fi{}

\makeatletter
\renewcommand\bibsection%
{
  \section*{\refname
    \@mkboth{\MakeUppercase{\refname}}{\MakeUppercase{\refname}}}
}
\makeatother
\setlength{\bibsep}{0pt}
\bibliographystyle{abbrvnat}
\bibliography{planar-supports}

\end{document}